\documentclass[conference]{IEEEtran}

\usepackage{graphicx}
\usepackage{amsfonts}
\usepackage{amsmath}
\usepackage{amssymb}

\newtheorem{thm}{Theorem}[section]

\newtheorem{lem}[thm]{Lemma}
\newtheorem{proof}[thm]{proof}

\newtheorem{defn}[thm]{Definition}
\newtheorem{rem}[thm]{Remark}
\newtheorem{exam}[thm]{Example}

\ifCLASSINFOpdf \else \fi 
\begin{document}
\title{Network Coding for $3$s$/n$t Sum-Networks}

\author{\IEEEauthorblockN{Wentu Song\IEEEauthorrefmark{1},
Chau Yuen\IEEEauthorrefmark{1}, Kai Cai\IEEEauthorrefmark{2}, and
Rongquan Feng\IEEEauthorrefmark{3}}
\IEEEauthorblockA{\IEEEauthorrefmark{1}Singapore University of
Technology and Design}
\IEEEauthorblockA{\IEEEauthorrefmark{2}Department of Mathematics,
University of Hong Kong}
\IEEEauthorblockA{\IEEEauthorrefmark{3}School of Mathematical
Sciences, Peking University, Peking, China\\ Emails:
wentu$\_$song@sutd.edu.sg; yuenchau@sutd.edu.sg;
eecaikai@gmail.com; fengrq@math.pku.edu.cn}}

\maketitle

\begin{abstract}
A sum-network is a directed acyclic network where each source
independently generates one symbol from a given field $\mathbb F$
and each terminal wants to receive the sum $($over $\mathbb F)$ of
the source symbols. For sum-networks with two sources or two
terminals, the solvability is characterized by the connection
condition of each source-terminal pair \cite{Rama08}. A necessary
and sufficient condition for the solvability of the $3$-source
$3$-terminal $(3$s$/3$t$)$ sum-networks was given by Shenvi and
Dey \cite{Shenvi10}. However, the general case of arbitrary
sources/sinks is still open. In this paper, we investigate the
sum-network with three sources and $n$ sinks using a region
decomposition method. A sufficient and necessary condition is
established for a class of $3$s$/n$t sum-networks. As a direct
application of this result, a necessary and sufficient condition
of solvability is obtained for the special case of $3$s$/3$t
sum-networks.
\end{abstract}


\IEEEpeerreviewmaketitle

\section{Introduction}
Network coding allows intermediate nodes of a communication
network to combine the incoming information before forwarding it,
and was shown to have significant throughput advantages as opposed
to the conventional store-and-forward scheme \cite{Ahlswede00,
Li03}.

Most of the existent works of network coding focus on how the
terminal nodes recover the whole or part of the original messages.
Recently, network coding for communicating the sum of source
messages to the terminal nodes was investigated
\cite{Rama08}-\cite{Rai13}. Such a network is called as a
sum-network. The problem of communicating sums over networks is in
fact a subclass of the problem of distributed function
computation, which has been considered in different contexts
\cite{Giridhar}-\cite{Kannan}.

It was shown in \cite{Rama08} that for directed acyclic graphs
with unit capacity edges and independent, unit-entropy sources, if
there are two sources or two terminals in the network, then the
network is solvable if and only if every source is connected to
every terminal. For the $3$-source $3$-terminal $(3$s$/3$t$)$
sum-networks, a necessary and sufficient condition for the
solvability over any field is given in \cite{Shenvi10}. However,
for networks with arbitrary number of sources and terminals, no
necessary and sufficient condition is known.

In this paper, we consider the sum-networks with three sources
using the technique of region decomposition
\cite{Wentu11,Wentu12}. We give a necessary and sufficient
condition for the solvability of a subclass of $3$s$/n$t
sum-networks. As a result, we give a simple characterization of
solvability for the special case of $3$s$/3$t sum-networks.

This paper is organized as follows. In Section
\uppercase\expandafter{\romannumeral 2}, we introduce the network
model and the notations.  The methodology is proposed in section
\uppercase\expandafter{\romannumeral 3}. The main result is
presented in Section \uppercase\expandafter{\romannumeral 4}. The
paper is concluded in Section \uppercase\expandafter{\romannumeral
5}.

\section{Models and Notations}
We consider a directed, acyclic, finite graph $G=(V,E)$ with a set
of $k$ sources $\{s_1,\cdots,s_k\}$ and a set of $n$ terminals
(sinks) $\{t_1,\cdots, t_n\}$. Each source $s_i$ generates a
message $X_i\in\mathbb F$ and each terminal $t_j$ wants to get the
sum $\sum_{i=1}^kX_i$, where $\mathbb F$ is a finite field. We
assume that each link is error-free, delay-free and can carry one
symbol from the field in each use. We call such network as a
$k$s$/n$t
sum-network. 

For a link $e=(u,v)\in E$, $u$ is called the \emph{tail} of $e$
and $v$ is called the \emph{head} of $e$, and are denoted by
$u=\text{tail}(e)$ and $v=\text{head}(e)$, respectively. We call
$e$ an incoming link of $v~($an outgoing link of $u)$. For two
links $e,e'\in E$, we call $e'$ an {\em incoming link} of $e$ ($e$
an {\em outgoing link} of $e'$) if
$\text{tail}(e)=\text{head}(e')$. For any $e\in E$, denoted by
$\text{In}(e)$ the set of incoming links of $e$.

To aid analysis, we assume that each source $s_i$ has an imaginary
incoming link, called the $X_{i}$ {\em source link} $($or a
\emph{source link} for short$)$, and each terminal $t_{j}$ has an
imaginary outgoing link, called a {\em terminal link}. Note that
the source links have no tail and the terminal links have no head.
As a result, the source links have no incoming link. For the sake
of convenience, if $e\in E$ is not a source link, we call $e$ a
\emph{non-source link}.

\renewcommand\figurename{Fig}
\begin{figure*}[htbp]
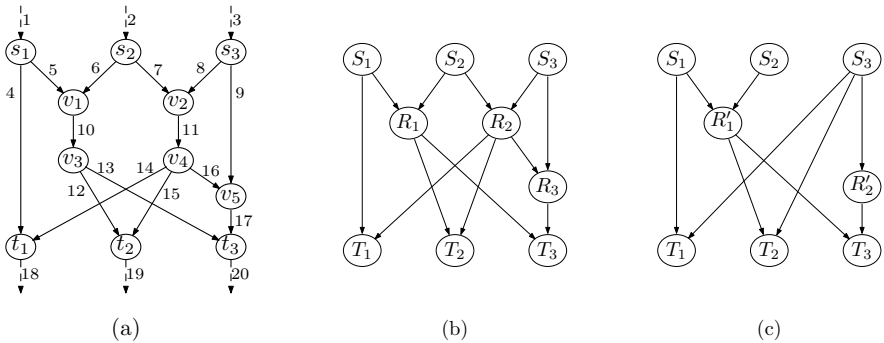

\begin{center}
\includegraphics[height=4.5cm]{3t-1.1}
\includegraphics[height=4cm]{3t-1.2}
\includegraphics[height=4cm]{3t-1.3}
\end{center}
\vspace{-0.4cm}\caption{Examples of region graph: (a) is a
$3$s$/3$t sum-network $G_1$, where all links are sequentially
indexed as $1,2,\cdots,20$. Here, the imaginary links $1,2,3$ are
the $X_1,X_2,X_3$ source link, and $18,19,20$ are the terminal
links at terminal $t_1,t_2,t_3$ respectively. (b) is the region
graph $\text{RG}(D)$, where $S_1=\{1,4,5\},S_2=\{2,6,7\},
S_3=\{3,8,9\}, R_1=\{10,12,13\}, R_2=\{11,14,15,16\}, R_3=\{17\},
T_1=\{18\}, T_2=\{19\}, T_3=\{20\}$ and
$D=\{S_1,S_2,S_3,R_1,R_2,R_3,T_1,T_2,T_3\}$. (c) is the region
graph $\text{RG}(D')$, where $S_1=\{1,4,5\},S_2=\{2,6,7\},
S_3=\{3,8,9\}, R_1'=\{10,12,13\}, R_2'=\{11,14,15,16,17\},
T_1=\{18\}, T_2=\{19\}, T_3=\{20\}$ and
$D'=\{S_1,S_2,S_3,R_1',R_2',T_1,T_2,T_3\}$.}\label{g-1}\vspace{-0.4cm}
\end{figure*}

Let $\mathbb F^k$ be the $k$-dimensional vector space over the
finite field $\mathbb F$. For any subset $A\subseteq\mathbb F^k$,
let $\langle A\rangle$ denote the subspace of $\mathbb F^k$
spanned by $A$. For $i\in\{1,\cdots,k\}$, we let $\alpha_i$ denote
the vector of $\mathbb F^k$ with the $i$th component being one and
all other components being zero. Meanwhile, we let
$\bar{\alpha}=\sum_{i=1}^k\alpha_i=(1,1,\cdots,1)$, i.e., the
vector with all components being one.

For any linear network coding scheme, the message along any link
$e$ is a linear combination $M_e=\sum_{i=1}^kc_iX_i$ of the source
messages and we use the corresponding coding vector
$d_e=(c_1,\cdots,c_k)$ to represent the message, where
$c_i\in\mathbb F$. To ensure the computability of network coding,
the outgoing message, as a $k$-dimensional vector, must be in the
span of all incoming messages. Moreover, to ensure that all
terminals receive the sum $\sum_{i=1}^kX_i$, if $e$ is a terminal
link of the sum-network, then
$d_e=\sum_{i=1}^k\alpha_i=\bar{\alpha}$. Thus, we can define a
linear network code of a $k$s$/n$t sum-network as follows:
\begin{defn}[Linear Network Code]\label{lnc}
Let $G=(V,E)$ be a $k$s$/n$t sum-network. A \emph{linear code}
(LC) of $G$ over the field $\mathbb F$ is a collection of vectors
$C=\{d_{e}\in\mathbb F^k; e\in E\}$ such that
\begin{itemize}
    \item[(1)] $d_{e}=\alpha_{i}$ if $e$ is the $X_i$ source link
    $(i=1,\cdots,k)$;
    \item[(2)] $d_e\in\langle d_{e'}; e'\in\text{In}(e)\rangle$
    if $e$ is a non-source link.
\end{itemize}
The code $C=\{d_{e}\in\mathbb F^k; e\in E\}$ is said to be a
\emph{linear solution} of $G$ if $d_{e}=\bar{\alpha}$ for all
terminal link $e$.
\end{defn}

The vector $d_e$ is called the {\em global encoding vector} of
link $e$. The network $G$ is said to be {\em solvable} if it has a
linear solution over some finite field $\mathbb F$.

\section{Region Decomposition and Network Coding}
In this section, we present the region decomposition approach,
which will take a key role in our discussion. The basic idea of
region decomposition is proposed in \cite{Wentu11,Wentu12}.

\subsection{Region Decomposition and Region Graph}
\begin{defn}[Region and Region Decomposition]\label{Reg}
Let $R$ be a non-empty subset of $E$. $R$ is called a region of
$G$ if there is an $e_{l}\in R$ such that for any $e\in R$ and
$e\neq e_{l}$, $R$ contains an incoming link of $e$. If $E$ is
partitioned into mutually disjoint regions, say
$R_{1},R_{2},\cdots,R_{N}$, then we call
$D=\{R_{1},R_{2},\cdots,R_{N}\}$ a region decomposition of $G$.
\end{defn}

The edge $e_{l}$ in Definition \ref{Reg} is called the leader of
$R$ and is denoted as $e_{l}=\text{lead}(R)$. A region $R$ is
called the $X_{i}$ \emph{source region} (or a \emph{source region}
for short) if $\text{lead}(R)$ is the $X_{i}$ source link; $R$ is
called a \emph{terminal region} if $R$ contains a terminal link.
If $R$ is neither a source region nor a terminal region, we call
$R$ a coding region. If $R$ is not a source region, we call $R$ a
\emph{non-source region}.

Since the source links have no incoming link, then each source
region contains exactly one source link, i.e., its leader. But a
terminal region may contains more than one terminal links. So
there are exactly $k$ source region and at most $n$ terminal
regions for any $k$s$/n$t sum-network. We will always denote the
$k$ source regions as $S_1,\cdots,S_k$ and the $n$ terminal
regions as $T_1,\cdots,T_n$.

\begin{defn}[Region Graph]\label{reg-g}
Let $D$ be a region decomposition of $G$. The region graph of $G$
about $D$ is a directed, simple graph with vertex set $D$ and edge
set $\mathcal E_D$, where $\mathcal E_D$ is the set of all ordered
pairs $(R',R)$ such that $R'$ contains an incoming link of
$\text{lead}(R)$.
\end{defn}

Consider the example network $G_1$ in Fig. \ref{g-1} (a). Examples
of two region graphs are shown in Fig. \ref{g-1} (b) and (c). In
general, $G$ may have many region decompositions.

We use $\text{RG}(D)$ to denote the region graph of $G$ about $D$,
i.e., $\text{RG}(D)=(D,\mathcal E_D)$. If $(R',R)$ is an edge of
$\text{RG}(D)$, we call $R'$ a parent of $R$. For $R\in D$, we use
$\text{In}(R)$ to denote the set of parents of $R$ in
$\text{RG}(D)$. Since the source links have no incoming link, then
the source regions have no parent. Moreover, since $G$ is acyclic,
then clearly, $\text{RG}(D)$ is acyclic.

For $R,R'\in D$, a path in $\text{RG}(D)$ from $R'$ to $R$ is a
sequence of regions $\{R_0=R',R_1,\cdots,R_p=R\}$ such that
$R_{i-1}$ is a parent of $R_i, i=1,\cdots,p$. If there is a path
from $R'$ to $R$, we say $R'$ is connected to $R$ and denote
$R'\rightarrow R$. Else, we say $R'$ is not connected to $R$ and
denote $R'\nrightarrow R$. In particular, we regard $R\rightarrow
R$ for all $R\in D$.

\subsection{Network Coding on Region Graph}
\begin{defn}[Codes on Region Graph]\label{lnc-reg-g}
A \emph{linear code} (LC) of the region graph $\text{RG}(D)$ over
the field $\mathbb F$ is a collection of vectors
$\tilde{C}=\{d_{R}\in\mathbb F^k; R\in D\}$ such that
\begin{itemize}
    \item[(1)] $d_{S_i}=\alpha_{i}$, where $S_i$ is the $X_i$ source
    region for each $i\in\{1,\cdots,k\}$;
    \item[(2)] $d_R\in\langle d_{R'}; R'\in\text{In}(R)\rangle$
    if $R$ is a non-source region.
\end{itemize}
The code $\tilde{C}=\{d_{R}\in\mathbb F^k; R\in D\}$ is said to be
a \emph{linear solution} of $\text{RG}(D)$ if
$d_{T_j}=\bar{\alpha}$ for each terminal region $T_j$.
\end{defn}

The vector $d_R$ is called the {\em global encoding vector} of
$R$. The region graph $\text{RG}(D)$ is said to be {\em feasible}
if it has a linear solution over some finite field $\mathbb F$.

By Definition \ref{lnc-reg-g}, for any linear solution of
$\text{RG}(D)$, it is always be that $d_{S_i}=\alpha_i$ and
$d_{T_j}=\bar{\alpha}$. So in order to obtain a solution, we only
need to specify the global encoding vector for each coding region.

Let $D$ be a region decomposition of $G$. Clearly, any linear
solution of $\text{RG}(D)$ can be extended to a linear solution of
$G$ by letting $d_e=d_R$ for each $R\in D$ and each $e\in R$. So
if $\text{RG}(D)$ is feasible, then $G$ is solvable. But
conversely, if $G$ is solvable, it is not necessary that
$\text{RG}(D)$ is feasible.

For the region graph $\text{RG}(D)$ in Fig. \ref{g-1} (b), let
$d_{R_1}=\alpha_1$ and $d_{R_2}=d_{R_3}=\alpha_2+\alpha_3$. Then
$\tilde{C}=\{d_R;R\in D\}$ is a linear solution of $\text{RG}(D)$
and we can obtain a linear solution of $G_1$ by letting $d_e=d_R$
for each $R\in D$ and each $e\in R$. However, the region graph
$\text{RG}(D')$ in Fig. \ref{g-1} (c) is not feasible because for
any linear code, by conditions (1), (2) of Definition
\ref{lnc-reg-g}, $d_{T_1}\in\langle\alpha_1,\alpha_3\rangle$. So
it is impossible that $d_{T_1}=\bar{\alpha}$.

In the following, we shall define a special region decomposition
$D^{**}$ of $G$, called the basic region decomposition of $G$,
which is unique and has the property that $G$ is solvable if and
only if the region graph $\text{RG}(D^{**})$ is feasible.

\begin{defn}[Basic Region Decomposition\cite{Wentu11}]\label{BRD}
Let $D^{**}$ be a region decomposition of $G$. $D^{**}$ is called
a basic region decomposition of $G$ if the following conditions
hold:
\begin{itemize}
    \item [(1)] For any $R\in D^{**}$ and any $e\in
    R\setminus\{\text{lead}(R)\}$, $\text{In}(e)\subseteq R$;
    \item [(2)] Each non-source region $R$ in $D^{**}$ has at least two
    parents in $\text{RG}(D^{**})$.
\end{itemize}
\end{defn}
Accordingly, the region graph $\text{RG}(D^{**})$ is called a
basic region graph of $G$.

For example, one can check that for the network $G_1$ in Fig.
\ref{g-1} (a), the region graph $\text{RG}(D)$ in Fig. \ref{g-1}
(b) a the basic region graph of $G_1$.

The basic region decomposition $D^{**}$ can be decided within time
$O(|E|)~($See Algorithm 5 in \cite{Wentu11}. Note that this
Algorithm can be generalized to networks with any $k$ sources
directely.$)$. The following two theorems were also derived in
\cite{Wentu11} $($See Theorem 4.4 and 4.5 of \cite{Wentu11}
respectively.$)$ and we omit their proof.

\begin{thm}\label{b-reg-unq} $G$ has a unique
basic region decomposition, hence has a unique basic region graph.
\end{thm}

\begin{thm}\label{solv-eqvlt}
$G$ is solvable if and only if $\text{RG}(D^{**})$ is feasible,
where $D^{**}$ is the basic region decomposition of $G$.
\end{thm}

\subsection{Super Region}
In this subsection, we always assume that $D$ is a region
decomposition of $G$ such that each non-source region has at least
two parents in $\text{RG}(D)$.

\begin{defn}[Super Region \cite{Wentu12}]\label{g-reg}
Let $D$ be a region decomposition of $G$ and
$\emptyset\neq\Theta\subseteq D$. The super region generated by
$\Theta$, denoted by $\text{reg}(\Theta)$, is defined recursively
as follows:
\begin{itemize}
    \item[(1)] If $R\in\Theta$, then $R\in\text{reg}(\Theta)$;
    \item[(2)] If $R\in D$ and $\text{In}(R)
    \subseteq\text{reg}(\Theta)$, then $R\in\text{reg}(\Theta)$.
\end{itemize}
\end{defn}
We define
$\text{reg}^\circ(\Theta)=\text{reg}(\Theta)\setminus\Theta$.
Moreover, if $\Theta=\{R_1,\cdots,R_k\}$, then we denote
$\text{reg}(\Theta)=\text{reg}(R_1,\cdots,R_k)$.

Since $\text{RG}(D)$ is acyclic, then $\text{reg}(\Theta)$ is well
defined.

\renewcommand\figurename{Fig}
\begin{figure}[htbp]
\begin{center}
\includegraphics[height=2.9cm]{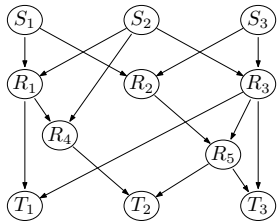}
\end{center}
\vspace{-0.2cm}\caption{An example of region graph.}\label{g-2}
\end{figure}

Consider the region graph in Fig. \ref{g-2}. We have
$\text{reg}(S_1,S_2)=\{S_1,S_2,R_1,R_4\}$. and
$\text{reg}(S_2,R_3)=\{S_2,S_3,R_3\}$.

\begin{rem}\label{rem-g-reg-code}
From Definition \ref{lnc-reg-g} and \ref{g-reg}, it is easy to see
that if $\tilde{C}=\{d_{R}\in\mathbb F^k; R\in D\}$ is a linear
code of $\text{RG}(D)$ and $\emptyset\neq\Theta\subseteq D$, then
$d_R\in\langle d_{R'};R'\in\Theta\rangle$ for all
$R\in\text{reg}(\Theta)$.
\end{rem}

\begin{lem}\label{g-reg-cap}
Suppose $\Theta_1$ and $\Theta_2$ are two subsets of $D$. Then
$\text{reg}(\Theta_1)\cap\text{reg}(\Theta_2)=\text{reg}(\Theta)$,
where
\begin{align*}
\Theta=(\text{reg}(\Theta_1)\cap\Theta_2)
\cup(\text{reg}(\Theta_2)\cap\Theta_1).
\end{align*}
\end{lem}
\begin{proof}
Clearly, $\Theta\subseteq\text{reg}(\Theta_1)$ and
$\Theta\subseteq\text{reg}(\Theta_2)$. Then by Definition
\ref{g-reg}, we have
$\text{reg}(\Theta)\subseteq\text{reg}(\Theta_1)
\cap\text{reg}(\Theta_2)$.

Now, suppose
$\text{reg}(\Theta_1)\cap\text{reg}(\Theta_2)\neq\text{reg}(\Theta)$.
Then there is an $R_0$ such that
$$R_0\in\text{reg}(\Theta_1)\cap\text{reg}(\Theta_2)
\backslash\text{reg}(\Theta).$$ By assumption of $\Theta$, we have
$R_0\notin\Theta_1\cup\Theta_2$. $($Otherwise, without loss of
generality, assume $R_0\in\Theta_1$. Then
$R_0\in(\text{reg}(\Theta_2)\cap\Theta_1)
\subseteq\Theta\subseteq\text{reg}(\Theta)$, which contradict to
the assumption of $R_0\notin\text{reg}(\Theta).)$ So
$R_0\in\text{reg}^\circ(\Theta_1)\cap\text{reg}^\circ(\Theta_2)$.
Then by Definition \ref{g-reg}, we have
$$\text{In}(R_0)\subseteq\text{reg}(\Theta_1)
\cap\text{reg}(\Theta_2).$$ Since $R_0\notin\text{reg}(\Theta)$,
then by Definition \ref{g-reg}, there exists an
$R_1\in\text{In}(R_0)$ such that $R_1\notin\text{reg}(\Theta)$.
Then
$$R_1\in\text{reg}(\Theta_1)\cap\text{reg}
(\Theta_2)\setminus\text{reg}(\Theta).$$ Similarly, $R_1$ has a
parent $R_2$ such that
$$R_2\in\text{reg}(\Theta_1)\cap\text{reg}
(\Theta_2)\setminus\text{reg}(\Theta).$$ By repeating this
process, we can find a series of infinite regions $R_0$, $R_1$,
$R_2$, $\cdots$ such that $R_i$ is a parent of $R_{i-1}$ and
$$R_i\in\text{reg}(\Theta_1)\cap\text{reg}
(\Theta_2)\setminus\text{reg}(\Theta),~ i=1, 2, \cdots.$$ This
contradicts to the fact that $\text{RG}(D)$ is a finite graph. So
$\text{reg}(\Theta_1)\cap\text{reg}(\Theta_2)=\text{reg}(\Theta)$.
\end{proof}

\section{A Sufficient and Necessary Condition for a Subclass of
$3$-source Sum-networks}

Throughout this section, we always assume that $G$ is a $3$s$/n$t
sum-network and $D^{**}$ is the basic region decomposition of $G$.
By Theorem \ref{solv-eqvlt}, $G$ is solvable if and only if
$\text{RG}(D^{**})$ is feasible. So we only need to consider
coding on $\text{RG}(D^{**})$.

Since $G$ is a $3$s$/n$t sum-network, then $\text{RG}(D^{**})$ has
exactly three source regions and at most $n$ terminal regions.
Without loss of generality, we assume $\text{RG}(D^{**})$ has
exactly $n$ terminal regions. Let $S_i ~(i\in\{1,2,3\})$ denote
the $X_i$ source region and $T_j,j=1,\cdots,n,$ denote the $n$
terminal regions. For any $i\in\{1,2,3\}$, by Lemma
\ref{g-reg-cap}, we have
\begin{align}
\text{reg}(S_{i},S_{j_1})\cap\text{reg}(S_{i},S_{j_2})=\{S_i\}
\label{eq-nt-1}
\end{align}
where $\{j_1,j_2\}=\{1,2,3\}\backslash\{i\}$. Thus
\begin{align}
\text{reg}^\circ(S_{i},S_{j_1})\cap\text{reg}^\circ(S_{i},S_{j_2})
=\emptyset.\label{eq-nt-2}
\end{align}
i.e., $\text{reg}^\circ(S_{1},S_{2})$,
$\text{reg}^\circ(S_{1},S_{3})$ and
$\text{reg}^\circ(S_{2},S_{3})$ are mutually disjoint. Denote
$\{1,\cdots,n\}=[n]$ for any positive integer $n$.
\begin{defn}\label{lmd-omd}
We define some subsets of $D^{**}$ as follows:
\begin{itemize}
    \item[\textbf{(1)}] $\Pi\triangleq\text{reg}(S_1,S_2)
    \cup\text{reg}(S_1,S_3)\cup\text{reg}(S_2,S_3)$;
    \item[\textbf{(2)}] For any $I\subseteq[n]$,
    $\Omega_I$ is the set of all $R\in D^{**}\backslash\Pi$ such
    that $R\rightarrow T_{j},\forall j\in I$, and
    $R\nrightarrow T_{j'},\forall j'\in[n]\setminus I$;
    \item[\textbf{(3)}] $\Lambda_I$ is the set of all $Q\in\Pi$
    such that $Q$ has a child $R\in\Omega_I$.
\end{itemize}
\end{defn}
We also denote $\Omega_I=\Omega_{i_1,\cdots,i_p}$ and
$\Lambda_I=\Lambda_{i_1, \cdots, i_p}$ if the subset
$I=\{i_1,\cdots,i_p\}$.

From the above definition, the following remark is obvious.
\begin{rem}\label{omg-intc}
If $I,I'\subseteq[n]$ and $I\neq I'$, then
$\Omega_{I}\cap\Omega_{I'}=\emptyset$.
\end{rem}

Since $\text{RG}(D^{**})$ is acyclic, regions in $D^{**}$ can be
sequentially indexed as $D^{**}=\{R_1,R_2,R_3,\cdots,R_N\}$ such
that $R_i=S_i, i=1,2,3$, $R_{N-n+j}=T_j, j=1,2,\cdots,n,$ and
$\ell <\ell'$ if $R_{\ell}$ is a parent of $R_{\ell'}$. For all
$I\subseteq[n]$, we can determine $\Omega_{I}$ by a the following
simple labelling algorithm.\vspace{-5mm}
\begin{center}
\setlength{\unitlength}{1mm}
\begin{picture}(90,-42)(0,35)
\put(0,0){\line(1,0){87}} \put(0,35){\line(1,0){87}}
\put(0,0){\line(0,1){35}} \put(87,0){\line(0,1){35}}
\end{picture}
\end{center}
\vspace{-0.05in} \noindent { \small \textbf{Algorithm 1}:
Labelling Algorithm $(\text{RG}(D^{**}))$:

\vspace{0.075in} $j\leftarrow$ from $1$ to $n$

\vspace{0.025in} ~ ~ Label $R_{N-n+j}$ with $j$;

\vspace{0.025in} $\ell\leftarrow$ from $N$ to $1$;

\vspace{0.025in} ~ ~ \textbf{if} $R_\ell$ has a child $R_{\ell'}$
such that $R_{\ell'}$ is labelled with $j$ for

\vspace{0.025in} ~ ~ some $j\in[n]$ \textbf{then}

\vspace{0.025in} ~ ~ ~ ~ Label $R_{\ell}$ with $j$;}
\vspace{0.55cm}

\vspace{0.05cm}Note that for each $R\in D^{**}$ and $j\in[n]$, if
$R$ has a child $R'$ such that $R'\rightarrow T_j$, then
$R\rightarrow T_j$. So $R\rightarrow T_j$ if and only if $R$ is
labelled with $j$ by Algorithm 1. Let $I_R=\{j\in[n]; R \text{~is
labelled with~} j\}$. Then for each $I\subseteq[n]$,
$R\in\Omega_I$ if and only if $I_R=I$. Thus, by Algorithm 1, we
can easily determine $\Omega_{I}$ for all $I\subseteq[n]$.
Clearly, the run time of Algorithm 1 is $O(|D^{**}|)$.

Consider the region graph in Fig. \ref{g-2}. We have
$\Omega_i=\{T_i\}$ for $i\in\{1,2,3\}$, $\Omega_{2,3}=\{R_5\}$ and
$\Omega_{1,2}=\Omega_{1,3}=\Omega_{1,2,3}=\emptyset$. Thus, we
have $\Lambda_1=\{R_1,R_3\}$, $\Lambda_2=\{R_4\}$,
$\Lambda_3=\{R_3\}$, $\Lambda_{2,3}=\{R_2,R_3\}$ and
$\Lambda_{1,2}=\Lambda_{1,3}=\Lambda_{1,2,3}=\emptyset$.

If $\tilde{C}=\{d_{R}\in\mathbb F^k; R\in D\}$ is a linear
solution of $\text{RG}(D^{**})$, $\{i,j\}\subseteq\{1,2,3\}$ and
$R\in\text{reg}(S_i,S_j)$, then by Remark \ref{rem-g-reg-code},
$\bar{\alpha}\neq d_R\in\langle\alpha_i,\alpha_j\rangle$. So
$T_\ell\notin\text{reg}(S_i,S_j)$ for all terminal region
$T_\ell$, which implies that $T_\ell\in
D^{**}\backslash\Pi,\forall \ell\in[n]$. For this reason, in this
section, we assume:

\textbf{\emph{Assumption} 1}: $T_j\notin \Pi$ for all $j\in[n]$.

\begin{defn}\label{t-sep} The region graph $\text{RG}(D^{**})$ is
said to be terminal-separable if $\Omega_{I}=\emptyset$ for all
$I\subseteq[n]$ such that $|I|>1$.
\end{defn}

For example, the region graph in Fig. \ref{g-3} is
terminal-separable. However, the region graph in Fig. \ref{g-2} is
not because $\Omega_{2,3}=\{R_5\}\neq\emptyset$.

We shall give a necessary and sufficient condition of feasibility
for terminal-separable region graph, by which it is easy to check
whether a terminal-separable region graph is feasible.

\begin{rem}\label{non-ts}
Terminal-separable region graphs is of interesting because, if a
region graph $\text{RG}(D)$ is not terminal-separable, then we can
view it as a region graph with fewer terminal regions. For
example, for the region graph in Fig. \ref{g-2}, we can view $T_1$
and $R_5$ as two terminal regions and construct a linear solution
of $\text{RG}(D)$. Then the sum  of sources can be transmitted
from $R_5$ to $T_2$ and $T_3$. In fact, let
$d_{R_1}=d_{R_2}=\alpha_1, d_{R_3}=\alpha_2+\alpha_3$ and
$d_{R_5}=\alpha_1+\alpha_2+\alpha_3$. Then $\{d_R;R\in D\}$ is a
linear solution of $\text{RG}(D)$.
\end{rem}

\begin{lem}\label{in-reg-lmd}
If $\text{RG}(D^{**})$ is terminal-separable, then for all
$j\in[n]$ and $\{i_1,i_2\}\subseteq\{1,2,3\}$, we have
$T_j\in\Omega_j\subseteq\text{reg}^\circ(\Lambda_j)$ and
$\Lambda_j\nsubseteq\text{reg}(S_{i_1},S_{i_2})$. In particular,
we have $|\Lambda_j|\geq 2$.
\end{lem}
\begin{proof}
Since $T_j\rightarrow T_j$ and $\text{RG}(D^{**})$ is
terminal-separable, then $T_j\nrightarrow T_{j'}, \forall
j'\in[n]\backslash\{j\}$. So $T_j\in\Omega_j$.

We now prove $\Omega_j\subseteq\text{reg}(\Lambda_j)$ by
contradiction. For this purpose, suppose there is an
$R\in\Omega_j$ such that $R\notin\text{reg}(\Lambda_j)$. Then by
Definition \ref{g-reg}, $R$ has a parent, say $P_1$, such that
$P_1\notin\text{reg}(\Lambda_j)$. Clearly, $P_1\notin\Pi$.
$($Otherwise, by the definition of $\Lambda_{j}$,
$P_1\in\Lambda_{j}\subseteq\text{reg}(\Lambda_j)$, which
contradicts to the assumption that
$P_1\notin\text{reg}(\Lambda_j)$.$)$ Since $\text{RG}(D^{**})$ is
terminal-separable and $P_1\rightarrow R\rightarrow T_j$, then
$P_1\nrightarrow T_{j'}, \forall j'\neq j$. So $P_1\in\Omega_j$.
Similarly, $P_1$ has a parent $P_2$ such that
$P_2\notin\text{reg}(\Lambda_j)$ and $P_2\in\Omega_{j}$. By
repeating this process, we can obtain a series of infinite
regions, $P_1,P_2,\cdots$ such that
$P_i\notin\text{reg}(\Lambda_j)$ and $P_i\in\Omega_{j}$, which
contradicts to the fact that $\text{RG}(D^{**})$ is a finite
graph. So $\Omega_j\subseteq\text{reg}(\Lambda_j)$.

Note that $\Omega_j\subseteq D^{**}\setminus\Pi$ and
$\Lambda_j\subseteq\Pi$. So $\Omega_j\cap\Lambda_j=\emptyset$.
Thus, we have
$T_j\in\Omega_j\subseteq\text{reg}^\circ(\Lambda_j)$.

Moreover, if $\Lambda_j\subseteq\text{reg}(S_{i_1},S_{i_2})$, then
by Definition \ref{g-reg}, we have
$T_j\in\Omega_j\subseteq\text{reg}^\circ(\Lambda_j)\subseteq
\text{reg}(S_{i_1},S_{i_2})$, which contradicts to Assumption 1.
So $\Lambda_j\nsubseteq\text{reg}(S_{i_1},S_{i_2})$.

Finally, if $|\Lambda_j|=1$, say $\Lambda_j=\{Q\}$, then by the
definition of $\Pi$ and $\Lambda_j$, we have
$Q\in\text{reg}(S_{i_1},S_{i_2})$ for some
$\{i_1,i_2\}\subseteq\{1,2,3\}$. Thus,
$\Lambda_j=\{Q\}\subseteq\text{reg}(S_{i_1},S_{i_2})$, which
contradicts to the proved result that
$\Lambda_j\nsubseteq\text{reg}(S_{i_1},S_{i_2})$. So
$|\Lambda_j|\geq 2$.
\end{proof}

\begin{lem}\label{lmd-solv}
Suppose $\text{RG}(D^{**})$ is terminal-separable. Then
$\text{RG}(D^{**})$ is feasible if and only if there is a
collection of vectors $\tilde{C}_\Pi=\{d_R;
R\in\Pi\}\subseteq\mathbb F^3$ satisfying the following three
conditions:
\begin{itemize}
  \item [(1)] $d_{S_i}=\alpha_i, i=1,2,3$;
  \item [(2)] $d_R\in\langle d_{R'}; R'\in\text{In}(R)\rangle,
  \forall R\in\Pi\setminus\{S_1,S_2,S_3\}$;
  \item [(3)] $\bar{\alpha}\in\langle d_R; R\in\Lambda_j\rangle,
  \forall j\in[n]$.
\end{itemize}
\end{lem}
\begin{proof}
Suppose $\text{RG}(D^{**})$ is feasible and $\tilde{C}=\{d_R; R\in
D^{**}\}$ is a linear solution of $\text{RG}(D^{**})$. By Lemma
\ref{in-reg-lmd},
$T_j\in\Omega_j\subseteq\text{reg}^\circ(\Lambda_j)$ and
$\Lambda_j\nsubseteq\text{reg}(S_{i_1},S_{i_2})$, $\forall j\in
[n]$ and $\{i_1,i_2\}\subseteq\{1,2,3\}$. By Remark
\ref{rem-g-reg-code}, $\bar{\alpha}=d_{T_j}\in\langle
d_R;R\in\Lambda_j\rangle$. Let $\tilde{C}_\Pi=\{d_R;R\in\Pi\}$.
Then $\tilde{C}_\Pi$ satisfies conditions (1)-(3).

Conversely, suppose there is a collection $\tilde{C}_\Pi=\{d_R;
R\in\Pi\}\subseteq\mathbb F^3$ satisfying conditions (1)-(3). We
can construct a linear solution of $\text{RG}(D^{**})$ as follows:

Since $\text{RG}(D^{**})$ is terminal-separable, for each
$j\in[n]$ and $Q\in\Lambda_j$, by the Definition of $\Lambda_j$
and $\Omega_j$, we can find a path
$\{R_1,\cdots,R_\ell\}\subseteq\Omega_j$ such that $R_\ell=T_j$
and $Q$ is a parent of $R_1$. Let $\Gamma_j$ be the union of all
such paths. Then $\Gamma_j\subseteq\Omega_j$. Since
$\bar{\alpha}\in\langle d_R; R\in\Lambda_j\rangle$, then we can
construct a code $\tilde{C}_{\Gamma_j}=\{d_R; R\in\Gamma_j\}$ such
that $d_{T_j}=\bar{\alpha}$ and $d_R\in\langle d_{R'};
R'\in\text{In}(R)\rangle$ for all $R\in\Gamma_j$. By Remark
\ref{omg-intc}, $\Omega_j,j=1,\cdots,n,$ are mutually disjoint. So
$\Gamma_j,j=1,\cdots,n,$ are mutually disjoint and
$\tilde{C}=\tilde{C}_\Pi\cup\tilde{C}_{\Gamma_1}\cup\cdots
\cup\tilde{C}_{\Gamma_n}$ is a linear solution of
$\text{RG}(D^{**})$. Thus, $\text{RG}(D^{**})$ is feasible.
\end{proof}

\subsection{Partitioning of $\Pi$}
To give a simple characterization of feasibility of
$\text{RG}(D^{**})$, we need to make some discussion on
partitioning $\Pi$.

Let $\mathcal I=\{\Delta_1,\cdots,\Delta_K\}$ be a partition of
$\Pi$. For the sake of convenience, we shall call each $\Delta_i$
an \emph{equivalent class} of $\mathcal I$. If $R\in\Delta_i$, we
denote $\Delta_i=[R]$. Thus, for each $\Delta_i$, we can choose an
$R_i\in\Delta_i$ and denote $\mathcal I=\{[R_1],\cdots,[R_K]\}$.

Let $\mathcal I=\{[S_1],[S_2],[S_3],\cdots,[R_K]\}$ be an
arbitrary partition of $\Pi$.\footnote{When we use the notation
$\mathcal I=\{[S_1],[S_2],[S_3],\cdots,[R_K]\}$, we always assume
that $[S_1],[S_2],[S_3],\cdots,[R_K]$ are mutually different.} For
each equivalent class $[R]\in\mathcal I$ and each subset
$\{i,j\}\subseteq\{1,2,3\}$, we denote
\begin{align} \vspace{-0.1cm}
[R]_{i,j}=[R]\cap\text{reg}(S_{i},S_{j}). \label{eq-nt-3}
\end{align}
For each $i\in\{1,2,3\}$, we denote
\begin{align} \vspace{-0.1cm}
[S_i]_i=[S_i]_{i,j_1}\cup[S_i]_{i,j_2} \label{eq-nt-4}
\end{align}
where $\{j_1,j_2\}=\{1,2,3\}\backslash\{i\}$. Then we can divide
each equivalent class as follows:
\begin{defn}\label{sub-class}
For $i\in\{1,2,3\}$, $[S_i]$ is divided into two \emph{subclasses}
$[S_i]_i$ and $[S_i]_{j_1,j_2}$, where
$\{j_1,j_2\}=\{1,2,3\}\backslash\{i\}$; For $i\in\{4,\cdots,K\}$,
$[R_i]$ is divided into three \emph{subclasses} $[R_i]_{1,2},
[R_i]_{1,3}$ and $[R_i]_{2,3}$.
\end{defn}

\renewcommand\figurename{Fig}
\begin{figure}[htbp]
\begin{center}
\includegraphics[height=3cm]{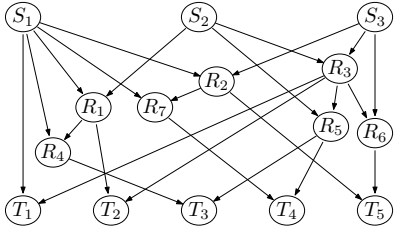}
\end{center}
\vspace{-0.2cm}\caption{An example of terminal-separable region
graph: By Definition \ref{g-reg}, we can check that
$\text{reg}(S_1,S_2)=\{S_1,S_2,R_1,R_4\}$,
$\text{reg}(S_1,S_3)=\{S_1,S_3,R_2,R_7\}$ and
$\text{reg}(S_2,S_3)=\{S_2,S_3,R_3,R_5,R_6\}$. By Definition
\ref{lmd-omd}, we have $\Omega_j=\{T_j\}, j=1,\cdots,6$ and
$\Omega_I=\emptyset, \forall I\subseteq\{1,\cdots,6\}$ such that
$|I|\geq 2$. So this region graph is terminal-separable.
}\label{g-3}\vspace{-0.05cm}
\end{figure}

For any equivalent class $[R]\in\mathcal I$, we use $[[R]]$ to
denote a subclass of $[R]$. Note that a subclass $[[R]]$ of $[R]$
is possibly an empty set. By Equations
(\ref{eq-nt-1})$-$(\ref{eq-nt-4}), each equivalent class is a
disjoint union of its all subclasses. Thus, $\{[S_1]_1,
[S_1]_{2,3}\}\cup\{[S_2]_2$, $[S_2]_{1,3}\}\cup\{[S_3]_3,
[S_3]_{1,2}\}\cup(\cup_{i=4}^K\{[R_i]_{1,2}, [R_i]_{1,3},
[R_i]_{2,3}\})$ is still a partition of $\Pi$.

\begin{exam}\label{ex-ptn}
Consider the region graph in Fig. \ref{g-3}. By Definition
\ref{g-reg}, $\text{reg}(S_1,S_2)=\{S_1,S_2,R_1,R_4\}$,
$\text{reg}(S_1,S_3)=\{S_1,S_3,R_2,R_7\}$ and
$\text{reg}(S_2,S_3)=\{S_2,S_3,R_3,R_5,R_6\}$. Let
$[S_1]=\{S_1,R_1,R_3,R_4,R_5,R_7\}$, $[S_2]=\{S_2\}$,
$[S_3]=\{S_3\}$, $[R_2]=\{R_2,R_6\}$ and $\mathcal
I=\{[S_1],[S_2],[S_3],[R_2]\}$. Then $\mathcal I$ is a partition
of $\Pi$ and $[S_1]_1=\{S_1,R_1,R_4,R_7\}$,
$[S_1]_{2,3}=\{R_3,R_5\}$, $[S_2]_2=\{S_2\}$,
$[S_2]_{1,3}=\emptyset$, $[S_3]_3=\{S_3\}$,
$[S_3]_{1,2}=\emptyset$, $[R_2]_{1,2}=\emptyset$,
$[R_2]_{1,3}=\{R_2\}$, $[R_2]_{2,3}=\{[R_6]\}$ are all subclasses
of $\mathcal I$ and they also form a partition of $\Pi$.
\end{exam}

\begin{defn}\label{cntd}
Let $\mathcal I=\{[S_1],[S_2],[S_3],\cdots,[R_K]\}$ be a partition
of $\Pi$. Two equivalent classes $[R']$ and $[R'']$ are said to be
\emph{connected} if one of the following conditions hold:
\begin{itemize}
  \item [(1)] There is a $j\in[n]$ such that
  $\Lambda_j\subseteq[[R']]\cup[[R'']]$, where $[[R']]~($resp.
  $[[R'']])$ is a subclass of $[R']~($resp. $[R''])$;
  \item [(2)] There is an $\{i_1,i_2\}\subseteq\{1,2,3\}$ such
  that $\text{reg}([R']_{i_1,i_2})\cap\text{reg}([R'']_{i_1,i_2})
  \neq\emptyset$.
\end{itemize}
\end{defn}

\begin{defn}\label{cmptl}
Let $\mathcal I=\{[S_1],[S_2],[S_3],\cdots,[R_K]\}$ be a partition
of $\Pi$. $\mathcal I$ is said to be \emph{compatible} if the
following two conditions hold:
\begin{itemize}
  \item [(1)] No pair of equivalent classes of $\mathcal I$
  are connected;
  \item [(2)] $\Lambda_j\nsubseteq[S_{i_1}]_{i_1} \cup[S_{i_2}]_{i_2}
  \cup(\cup_{\ell=4}^K[R_\ell]_{i_1,i_2})$ for all $j\in[n]$ and
  $\{i_1,i_2\} \subseteq\{1,2,3\}$.
\end{itemize}
\end{defn}

Clearly,  in Example \ref{ex-ptn}, the partition $\mathcal I$ of
$\Pi$ is compatible.

Suppose $\mathcal I$ is a partition of $\Pi$ and $\{[R'],
[R'']\}\subseteq\mathcal I$. By combining $[R']$ and $[R'']$ into
one equivalent class $[R']\cup[R'']$, we obtain a partition
$\mathcal I'=\mathcal I\cup\{[R']\cup[R'']\}\setminus\{[R'],
[R'']\}$ of $\Pi$. We call $\mathcal I'$ a \emph{contraction} of
$\mathcal I$ by combining $[R']$ and $[R'']$.

\subsection{Main Result}
Let $\mathcal I_0=\{[R]; R\in\Pi\}$, where $[R]=\{R\}, \forall
R\in\Pi$. Then $\mathcal I_0$ is a partition of $\Pi$. We call
$\mathcal I_0$ the trivial partition of $\Pi$.

\begin{defn}\label{chpt}
Let $\mathcal I_0,\mathcal I_1,\cdots,\mathcal I_L=\mathcal I_c$
be a sequence of partitions of $\Pi$ such that $\mathcal I_\ell$
is a contraction of $\mathcal I_{\ell-1}$ by combining two
connected equivalent classes in $\mathcal I_{\ell-1}$ and, for any
$\{i,j\}\subseteq\{1,2,3\}$, $[S_i]\neq[S_j]$ in $\mathcal
I_{\ell-1}$, where $\ell=1,\cdots,L$. $\mathcal I_c$ is called a
character partition of $\Pi$ if one of the following conditions
hold:
\begin{itemize}
  \item [(1)] $[S_i]=[S_j]$ for some $\{i,j\}\subseteq\{1,2,3\}$;
  \item [(2)] No pair of equivalent classes in $\mathcal I_c$ are
  connected.
\end{itemize}
\end{defn}

\begin{exam}\label{ex-infsb}
Consider the region graph in Fig. \ref{g-4} (a). We have
$\Pi=\{S_1,S_2,S_3,R_1,R_2,R_3,R_4\}$ and
$\Lambda_2=\{S_2,R_1\}\subseteq([S_2]_2\cup[R_1]_{1,3})$. By (1)
of Definition \ref{cntd}, $[S_2]$ and $[R_1]$ are connected. So
$\mathcal I_1=\{\{S_1\},\{S_2,R_1\},\{S_3\},\{R_2\}$,
$\{R_3\},\{R_4\}\}$ is obtained from $\mathcal I_0$ by combining
$[S_2]$ and $[R_1]$, where $\mathcal I_0$ is the trivial partition
of $\Pi$. Similarly, let $\mathcal
I_2=\{\{S_1\},\{S_2,R_1,R_2\},\{S_3\},\{R_3\},\{R_4\}\}$ and
$\mathcal I_3=\{\{S_1\},\{S_2,R_1,R_2,R_3\},\{S_3\},\{R_4\}\}$.
Then $\mathcal I_j, j=2,3,$ is obtained from $\mathcal I_{j-1}$ by
combining two connected equivalent classes. Note that in $\mathcal
I_3$,
$\text{reg}([S_2]_{2,3})=\text{reg}(R_2,R_3)=\{R_2,R_3,R_4\}$ and
$\text{reg}([R_4]_{2,3})=\text{reg}(R_4)=\{R_4\}$. So by (2) of
Definition \ref{cntd}, $[S_2]$ and $[R_4]$ are connected and
$\mathcal I_4=\{\{S_1\}$, $\{S_2,R_1,R_2,R_3,R_4\},\{S_3\}\}$ is
obtained from $\mathcal I_{3}$ by combining two connected
equivalent classes. In $\mathcal I_{4}$, again by (1) of
Definition \ref{cntd}, $[S_2]$ and $[S_1]$ are connected and
$\mathcal I_5=\{\{S_1,S_2,R_1,R_2,R_3$, $R_4\}$, $\{S_3\}\}$ is
obtained from $\mathcal I_{4}$ by combining two connected
equivalent classes. Thus, $\mathcal I_0,\mathcal
I_1,\cdots,\mathcal I_5$ satisfy the conditions of Definition
\ref{chpt}. So $\mathcal I_c=\mathcal I_5$ is a character
partition of $\Pi$. Since $[S_1]=[S_2]$, then $\mathcal I_c$ is
not compatible.

Similarly, for the region graph in Fig. \ref{g-4} (b), we can find
that $\mathcal I_c=\{\{S_1,P_1,P_2,P_3,P_4\},\{S_2\},\{S_3\}\}$ is
a character partition of $\Pi$. Since $\Lambda_2\subseteq[S_1]_1$,
then $\mathcal I_c$ is not compatible.
\end{exam}

\renewcommand\figurename{Fig}
\begin{figure}[htbp]
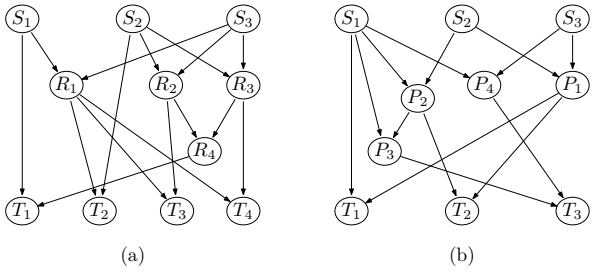

\begin{center}
\includegraphics[height=3.5cm]{fsb.2}
\includegraphics[height=3.5cm]{fsb.3}
\end{center}
\vspace{-0.2cm}\caption{Examples of region
graph.}\label{g-4}\vspace{-0.4cm}
\end{figure}

\begin{lem}\label{comp-code}
Let $\mathcal I$ be a partition of $\Pi$. If $\mathcal I$ is
compatible, then $\text{RG}(D^{**})$ is feasible.
\end{lem}
\begin{proof}
The proof is given in Appendix A.
\end{proof}

\begin{lem}\label{chpt-lem}
Suppose $\tilde{C}_\Pi=\{d_R; R\in\Pi\}\subseteq\mathbb F^3$
satisfies the conditions of Lemma \ref{lmd-solv} and $\mathcal
I_c$ is a character partition of $\Pi$. For any $[R]\in\mathcal
I_c$ and $\{i_1,i_2\}\subseteq\{1,2,3\}$, if $Q\in[R]_{i_1,i_2}$
and $d_{Q}\neq 0$, then $d_{Q'}\in\langle d_{Q}\rangle, \forall
Q'\in[R]_{i_1,i_2}$.
\end{lem}
\begin{proof}
The proof is given in Appendix B.
\end{proof}

\begin{thm}\label{lc-md}
Let $\text{RG}(D^{**})$ be terminal-separable and $\mathcal I_c$
be a character partition of $\Pi$. Then $\text{RG}(D^{**})$ is
feasible if and only if $\mathcal I_c$ is compatible. Moreover, it
is $\{|\Pi|,n\}$-polynomial time complexity to determine
feasibility of $\text{RG}(D^{**})$.
\end{thm}
\begin{proof}
If $\mathcal I_c$ is compatible, then by Lemma \ref{comp-code},
$\text{RG}(D^{**})$ is feasible. Conversely, suppose
$\text{RG}(D^{**})$ is feasible and $\tilde{C}_\Pi=\{d_R;
R\in\Pi\}\subseteq\mathbb F^3$ satisfies the conditions of Lemma
\ref{lmd-solv}. We shall prove that $\mathcal I_c$ is compatible.

For any $\{i_1,i_2\}\subseteq\{1,2,3\}$, if $[S_{i_1}]=[S_{i_2}]$,
then $S_{i_1}\in[S_{i_2}]_{i_1,i_2}$. By Definition
\ref{lnc-reg-g} and Lemma \ref{chpt-lem},
$d_{S_{i_1}}=\alpha_{i_1}\in\langle
d_{S_{i_2}}\rangle=\langle\alpha_{i_2}\rangle$, a contradiction.
So $[S_{i_1}]\neq[S_{i_2}]$. Thus, by proper naming, we can assume
$\mathcal I_c=\{[S_1],[S_2],[S_3],[R_4],\cdots,[R_K]\}$. Moreover,
by Definition \ref{chpt}, no pair of equivalent classes of
$\mathcal I_c$ are connected.

For any $j\in[n]$ and $\{i_1,i_2\}\subseteq\{1,2,3\}$, suppose
$\Lambda_j\subseteq[S_{i_1}]_{i_1} \cup[S_{i_2}]_{i_2}
\cup(\cup_{\ell=4}^K[R_\ell]_{i_1,i_2})$. Then by Lemma
\ref{chpt-lem} and Equation (\ref{eq-nt-4}), we have
\begin{align}
d_{Q}\in\langle d_{S_{i_1}}\rangle=\langle\alpha_{i_1}\rangle,
\forall Q\in[S_{i_1}]_{i_1} \label{eq-lc-md-1}
\end{align}
and
\begin{align}
d_{Q}\in\langle d_{S_{i_2}}\rangle=\langle\alpha_{i_2}\rangle,
\forall Q\in[S_{i_1}]_{i_2}. \label{eq-lc-md-2}
\end{align}
Note that $\tilde{C}_\Pi$ satisfies condition (2) of Lemma
\ref{lmd-solv}. By Equation (\ref{eq-lc-md-1}), (\ref{eq-lc-md-2})
and Definition \ref{g-reg}, we can easily see that
$d_R\in\langle\alpha_{i_1},\alpha_{i_2}\rangle$ for all
$R\in[S_{i_1}]_{i_1} \cup[S_{i_2}]_{i_2}
\cup(\cup_{\ell=4}^K[R_\ell]_{i_1,i_2})$. Then
$d_R\in\langle\alpha_{i_1},\alpha_{i_2}\rangle$ for all
$R\in\Lambda_j$ and $\bar{\alpha}\notin\langle
d_R;R\in\Lambda_j\rangle$, which contradicts to the assumption
that $\tilde{C}_\Pi$ satisfies condition (3) of Lemma
\ref{lmd-solv}. Thus, $\Lambda_j\nsubseteq[S_{i_1}]_{i_1}
\cup[S_{i_2}]_{i_2}\cup(\cup_{\lambda=4}^K[R_\lambda]_{i_1,i_2})$.
By Definition \ref{cmptl}, $\mathcal I_c$ is compatible.

By Definition \ref{chpt}, the following algorithm output a
character partition of $\Pi$.\vspace{-1.0cm}
\begin{center}
\setlength{\unitlength}{1mm}
\begin{picture}(90,-40)(0,58)
\put(0,0){\line(1,0){87}} \put(0,54){\line(1,0){87}}
\put(0,0){\line(0,1){54}} \put(87,0){\line(0,1){54}}
\end{picture}
\end{center}
\vspace{0.1in} \noindent { \small \textbf{Algorithm 2}:
Partitioning algorithm $(\Pi,\mathcal S)$:

\vspace{0.075in} $L=0$;

\vspace{0.025in} \textbf{While} there are $R',R''\in\mathcal I_L$
which are $\mathcal S-$connected \textbf{do}

\vspace{0.025in} ~ ~ Let $\mathcal I_{L+1}$ be a contraction of
$\mathcal I_{L}$ by combining $R'$ and $R''$;

\vspace{0.025in} ~ ~ \textbf{If} $[S_i]=[S_j]$ for some
$\{i,j\}\subseteq\{1,2,3\}$ \textbf{then}

\vspace{0.025in} ~ ~ ~ ~ $\mathcal I_c=\mathcal I_L$;

\vspace{0.025in} ~ ~ ~ ~ \textbf{return} $\mathcal I_c$;

\vspace{0.025in} ~ ~ ~ ~ \textbf{stop};

\vspace{0.025in} ~ ~ \textbf{else}

\vspace{0.025in} ~ ~ ~ ~ $L=L+1$;

\vspace{0.025in} $\mathcal I_c=\mathcal I_L$;

\vspace{0.025in} \textbf{return} $\mathcal I_c$; }

\vspace{0.45cm} Clearly, there are at most $|\Pi|$ rounds in
Algorithm 2 before output $\mathcal I_c$. In each round, we need
to determine wether there are two $\mathcal S$-connected
equivalent classes, which can be done in time $O(|\mathcal
S|)=O(n)$ by Definition \ref{cntd}. Thus, it is
$\{|\Pi|,n\}$-polynomial time complexity to determine whether
$\text{RG}(D^{**})$ is feasible.
\end{proof}

Consider the region graph in Fig. \ref{g-3}. We can check that the
partition $\mathcal I$ in Example \ref{ex-ptn} is a character
partition of $\Pi$. Since $\mathcal I$ is compatible, so the
region graph is feasible. Let $\mathbb F=GF(p)$ for a sufficiently
large prime $p$. Let $d_{R_1}=d_{R_4}=d_{R_7}=\alpha_1$,
$d_{R_2}=2\alpha_1+3\alpha_3$,
$d_{R_3}=d_{R_5}=\alpha_2+\alpha_3$, $d_{R_6}=\alpha_1+3\alpha_2$.
Then $\{d_R;R\in\Pi\}$ is a linear solution of the graph.

Similar to the information flow decomposition technique used in
\cite{Fragouli06}, we can reduce any compatible partition of $\Pi$
into a minimal compatible partition $\mathcal I_m$, i.e.,
$\mathcal I_m$ is a compatible partition of $\Pi$ but any
contraction of $\mathcal I_m$ is not compatible. Then we can
construct an optimal linear solution of $\text{RG}(D^{**})$ on
$\mathcal I_m$ using the method in the proof of Lemma
\ref{comp-code}.

For the two region graphs in Fig. \ref{g-4}, we have seen that
there is a character partition of $\Pi$ that is not compatible. So
by Theorem \ref{lc-md}, these two region graphs are not feasible.

In \cite{Shenvi10}, a necessary and sufficient condition for
solvability of a $3$s$/3$t sum-network was given based on a set of
connection conditions. By our method, we can give another
sufficient and necessary condition for solvability of $3$s$/3$t
sum-networks which is different from [6]:
\begin{thm}\label{3s-3t}
Suppose $\text{RG}(D^{**})$ has three terminal regions. Then
$\text{RG}(D^{**})$ is not feasible if and only if it is terminal
separable and the following condition (C-IR) hold:
\\ (\textbf{C-IR}) By proper naming, there is a
$P_1\in\text{reg}^\circ(S_2,S_3)$ and a
$P_2\in\text{reg}^\circ(S_1,S_2)$ such that
$\Lambda_1=\{S_1,P_1\}$, $\Lambda_2=\{P_1,P_2\}$ and
$\Lambda_3\subseteq\text{reg}(S_1,P_2)\cup\text{reg}(S_1,S_3)$.
\end{thm}
\begin{proof}
The proof is given in Appendix C.
\end{proof}

Fig. \ref{g-4} (b) is an illustration of infeasible region graph
of $3$s$/3$t sum-network.

\section{Conclusions and Discussions}
We investigated the network coding problem of a special subclass
of $3$s$/n$t sum-networks termed as terminal-separable networks
using a network region decomposition method. We give a necessary
and sufficient condition for solvability of terminal separable
networks as well as a simple characterization of solvability of
$3$s$/3$t sum-networks. The region decomposition method is shown
to be an efficient tool for analyzing the structure of a network
and helps to investigate the network coding problem of a
communication network. By more intensive analysis, we can also
give a characterization of solvability of $3$s$/4$t sum-networks,
which is our future work.

\appendices
\section{Proof of Lemma \ref{comp-code}}

Here, we prove Lemma \ref{comp-code}. First, we prove two lemmas.

\begin{lem}\label{genc-code}
Let $\mathcal B_1=\{\alpha_1,\alpha_2+\alpha_3\}$, $\mathcal
B_2=\{\alpha_2,\alpha_1+\alpha_3\}$, $\mathcal
B_3=\{\alpha_3,\alpha_1+\alpha_2\}$ and $K\geq 3$ is an integer.
If $\mathbb F$ is sufficiently large, then there are $K-3$ subsets
$\mathcal
B_4=\{\beta^{(4)}_{1,2},\beta^{(4)}_{1,3},\beta^{(4)}_{2,3}\}$,
$\cdots$, $\mathcal
B_K=\{\beta^{(K)}_{1,2},\beta^{(K)}_{1,3},\beta^{(K)}_{2,3}\}\subseteq\mathbb
F^3$ such that $\{\mathcal B_1,\mathcal B_2,\mathcal
B_3,\cdots,\mathcal B_K\}$ satisfies the following conditions:
\begin{itemize}
  \item [(1)] For any $\ell\in\{4,\cdots,K\}$ and $\{i_1,i_2\}
  \subseteq\{1,2,3\}$, $\beta^{(\ell)}_{i_1,i_2}\in\langle\alpha_{i_1},
  \alpha_{i_2}\rangle$;
  \item [(2)] For any $\ell\in\{1,\cdots,K\}$ and $\{\gamma,
  \gamma'\}\subseteq\mathcal B_\ell$, $\bar{\alpha}\in\langle\gamma,
  \gamma'\rangle$;
  \item [(3)] If $\{\gamma, \gamma',\gamma''\}\subseteq
  \cup_{\ell=1}^K\mathcal B_\ell$ such that $\{\gamma, \gamma',\gamma''\}
  \nsubseteq\langle\alpha_{i_1},\alpha_{i_2}\rangle, \forall\{i_1,i_2\}
  \subseteq\{1,2,3\}$, and $\{\gamma,\gamma',\gamma''\}
  \neq\{\beta^{(\ell)}_{1,2},\beta^{(\ell)}_{1,3},\beta^{(\ell)}_{2,3}\},
  \forall \ell\in\{4,\cdots,K\}$, then $\gamma, \gamma'$ and $\gamma''$
  are linearly independent;
  \item [(4)] For any pair $\{\gamma, \gamma'\}\subseteq
  \cup_{\ell=1}^K\mathcal
  B_\ell$, $\gamma$ and $\gamma'$ are linearly independent.
\end{itemize}
\end{lem}
\begin{proof}
We can prove this lemma by induction.

Clearly, when $K=3$, the collection $\{\mathcal B_1, \mathcal B_2,
\mathcal B_3\}$ satisfies conditions (1)$-$(4).

Now suppose $K>3$ and there is a collection $\{\mathcal B_1$,
$\cdots$, $\mathcal B_{K-1}\}$ which satisfies conditions
(1)$-$(4). We want to construct a subset $\mathcal
B_{K}=\{\beta^{(K)}_{1,2},\beta^{(K)}_{1,3},\beta^{(K)}_{2,3}\}
\subseteq\mathbb F^3$ such that the collection $\{\mathcal B_1,
\cdots, \mathcal B_{K-1}, \mathcal B_K\}$ satisfies conditions
(1)$-$(4). The subset $\mathcal B_{K}$ can be constructed as
follows:

Let $\Phi_{K-1}$ be the set of all pairs
$\{\gamma,\gamma'\}\subseteq\cup_{\ell=1}^{K-1}\mathcal B_\ell$
such that
$\{\gamma,\gamma'\}\nsubseteq\langle\alpha_{i_1},\alpha_{i_2}\rangle,
\forall\{i_1,i_2\}\subseteq\{1,2,3\}$. Then
$\langle\gamma,\gamma'\rangle
\cap\langle\alpha_{i_1},\alpha_{i_2}\rangle$ is an 1-dimensional
subspace of $\mathbb F^3$. Let
$\langle\gamma,\gamma'\rangle_{i_1,i_2}$ be a fixed non-zero
vector in $\langle\gamma,\gamma'\rangle
\cap\langle\alpha_{i_1},\alpha_{i_2}\rangle$. Let
\begin{align}
\Psi_{K-1}=\bigcup_{\{\gamma,\gamma'\}\in\Phi_{K-1}}\{\langle\gamma,
\gamma'\rangle_{1,2}, \langle\gamma,\gamma'\rangle_{1,3},
\langle\gamma,\gamma'\rangle_{2,3}\}. \notag 
\end{align}
Since $\mathbb F$ is sufficiently large, then there exists a
$\beta^{(K)}\in\mathbb F^3$ such that
\begin{align}
\beta^{(K)}\notin\langle\bar{\alpha}, \gamma\rangle,
\forall\gamma\in\Psi_{K-1}.\label{eq-genc-code-1}
\end{align}
For each $\{i_1,i_2\}\subseteq\{1,2,3\}$, let
\begin{align}
0\neq\beta^{(K)}_{i_1,i_2}\in\langle\beta^{(K)},
\bar{\alpha}\rangle\cap\langle\alpha_{i_1},
\alpha_{i_2}\rangle\label{eq-genc-code-2}
\end{align}
where $0$ is the zero vector of $\mathbb F^3$. Let $\mathcal
B_{K}=\{\beta^{(K)}_{1,2}$, $\beta^{(K)}_{1,3},
\beta^{(K)}_{2,3}\}$. We shall prove that the collection
$\{\mathcal B_1$, $\cdots$, $\mathcal B_{K-1}$, $\mathcal B_K\}$
satisfies conditions (1)$-$(4).

By Equation (\ref{eq-genc-code-2}), we have
$\beta^{(K)}_{i_1,i_2}\in\langle\alpha_{i_1}, \alpha_{i_2}\rangle,
\forall\{i_1,i_2\}\subseteq\{1,2,3\}$. So $\{\mathcal B_1$,
$\cdots$, $\mathcal B_{K-1}$, $\mathcal B_K\}$ satisfies condition
(1).

By assumption, $\{\mathcal B_1$, $\cdots$, $\mathcal B_{K-1}\}$
satisfies condition (2), then for any $\ell\in\{1,\cdots,K-1\}$
and $\{\gamma,\gamma'\}\subseteq\mathcal B_\ell$, the pair
$\{\gamma,\gamma'\}$ is in $\Phi_{K-1}$. Moreover, since
$\{\mathcal B_1$, $\cdots$, $\mathcal B_{K-1}\}$ satisfies
condition (1), then $\{\gamma,\gamma'\}\subseteq\mathcal
B_\ell\subseteq\langle\alpha_{1},\alpha_{2}\rangle\cup\langle\alpha_{1},
\alpha_{3}\rangle\cup\langle\alpha_{2},\alpha_{3}\rangle$. So
$\{\gamma,\gamma'\}\subseteq\{\langle\gamma,
\gamma'\rangle_{1,2}$, $\langle\gamma,\gamma'\rangle_{1,3}$,
$\langle\gamma,\gamma'\rangle_{2,3}\}$. Thus, we have
$\cup_{\ell=1}^{K-1}\mathcal B_\ell\subseteq\Psi_{K-1}$. By
Equation (\ref{eq-genc-code-1}), for any
$\gamma\in\cup_{\ell=1}^{K-1}\mathcal B_\ell$,
\begin{align}
\gamma\notin\langle\beta^{(K)},\bar{\alpha}\rangle.\label{eq-genc-code-4}
\end{align}
In particular, we have
$\alpha_{j}\notin\langle\beta^{(K)},\bar{\alpha}\rangle, j=1,2,3$.
So by Equation (\ref{eq-genc-code-2}), $\beta^{(K)}_{1,2}$,
$\beta^{(K)}_{1,3}$ and $\beta^{(K)}_{2,3}$ are mutually linearly
independent and $\bar{\alpha}\in\langle\gamma,\gamma'\rangle,
\forall\{\gamma, \gamma'\}\subseteq\mathcal B_K$. Thus,
$\{\mathcal B_1$, $\cdots$, $\mathcal B_{K-1},$ $\mathcal B_{K}\}$
satisfies condition (2).

Now, we prove that $\cup_{\ell=1}^{K}\mathcal B_\ell$ satisfies
condition (3). Suppose
$\{\gamma,\gamma',\gamma''\}\subseteq\cup_{\ell=1}^{K}\mathcal
B_\ell$ such that $\{\gamma, \gamma',\gamma''\}
\nsubseteq\langle\alpha_{i_1},\alpha_{i_2}\rangle$ for any
$\{i_1,i_2\}\subseteq\{1,2,3\}$ and $\{\gamma, \gamma', \gamma''\}
\neq\{\beta^{(\ell)}_{1,2}$, $\beta^{(\ell)}_{1,3}$,
$\beta^{(\ell)}_{2,3}\}$ for any $\ell\in\{1,\cdots,K\}$. We have
the following three cases:

Case 1:
$\{\gamma,\gamma',\gamma''\}\subseteq\cup_{\ell=1}^{K-1}\mathcal
B_\ell$. By the induction assumption, $\gamma, \gamma'$ and
$\gamma''$ are linearly independent.

Case 2: $\{\gamma,\gamma'\}\subseteq\cup_{\ell=1}^{K-1}\mathcal
B_\ell$ and $\gamma''\in\mathcal B_K$. We have the following two
subcases:

Case 2.1: $\{\gamma,\gamma'\}\subseteq\langle\alpha_{\ell_1},
\alpha_{\ell_2}\rangle$ for some
$\{\ell_1,\ell_2\}\subseteq\{1,2,3\}$. By assumption of $\{\gamma,
\gamma',\gamma''\}$, we have
$\gamma''\notin\langle\alpha_{\ell_1}, \alpha_{\ell_2}\rangle$. So
$\gamma, \gamma'$ and $\gamma''$ are linearly independent.

Case 2.2: $\{\gamma,\gamma'\}\nsubseteq\langle\alpha_{\ell_1},
\alpha_{\ell_2}\rangle, \forall
\{\ell_1,\ell_2\}\subseteq\{1,2,3\}$. Then the pair
$\{\gamma,\gamma'\}$ is in the set $\Phi_{K-1}$. So we have
$\gamma''\notin\langle\gamma,\gamma'\rangle. ~($Otherwise,
$\gamma''\in\{\langle\gamma, \gamma'\rangle_{1,2},
\langle\gamma,\gamma'\rangle_{1,3},
\langle\gamma,\gamma'\rangle_{2,3}\}\subseteq\Psi_{K-1}$ and by
Equation (\ref{eq-genc-code-2}),
$\beta^{(K)}\in\langle\bar{\alpha}, \gamma''\rangle$, which
contradicts to Equation (\ref{eq-genc-code-1}).$)$ Thus, $\gamma,
\gamma'$ and $\gamma''$ are linearly independent.

Case 3: $\gamma\in\cup_{\ell=1}^{K-1}\mathcal B_\ell$ and
$\{\gamma',\gamma''\}\subseteq\mathcal B_K$. By Equations
(\ref{eq-genc-code-2}) and (\ref{eq-genc-code-4}),
$\gamma\notin\langle\beta^{(K)},\bar{\alpha}\rangle=
\langle\gamma',\gamma''\rangle$. So $\gamma, \gamma'$ and
$\gamma''$ are linearly independent.

Thus, $\{\mathcal B_1,\cdots,\mathcal B_{K-1},\mathcal B_K\}$
satisfies conditions (3).

Clearly, if $\{\mathcal B_1,\cdots,\mathcal B_{K-1},\mathcal
B_K\}$ satisfies conditions (3), then for any $\{\gamma,
\gamma'\}\subseteq\cup_{\ell=1}^K\mathcal B_\ell$, we can find a
$\gamma''\in\cup_{\ell=1}^K\mathcal B_\ell$ such that $\gamma,
\gamma'$ and $\gamma''$ are linearly independent. So $\gamma$ and
$\gamma'$ are linearly independent and $\{\mathcal
B_1,\cdots,\mathcal B_{K-1},\mathcal B_K\}$ satisfies conditions
(4).

By induction, for all $K\geq 3$, we can always find a collection
$\{\mathcal B_1$, $\cdots$, $\mathcal B_{K-1}$, $\mathcal B_K\}$
which satisfies conditions (1)$-$(4).
\end{proof}

We give an example of Lemma \ref{genc-code} in the below. To
simplify our discussion, we assume that $\mathbb F=GF(p)$, where
$p$ is a sufficiently large prime.
\begin{exam}\label{ex-genc-code}
According to Lemma \ref{genc-code}, $\mathcal B_1=\{\alpha_1,
\alpha_2+\alpha_3\}, \mathcal B_2=\{\alpha_2, \alpha_1+\alpha_3\},
\mathcal B_3=\{\alpha_3, \alpha_1+\alpha_2\}$. Then
$\Phi_3=\{\{\alpha_1, \alpha_2+\alpha_3\}, \{\alpha_2,
\alpha_1+\alpha_3\}, \{\alpha_3, \alpha_1+\alpha_2\},
\{\alpha_2+\alpha_3, \alpha_1+\alpha_3\}, \{\alpha_2+\alpha_3,
\alpha_1+\alpha_2\}, \{\alpha_1+\alpha_3, \alpha_1+\alpha_2\}\}$
and $\Psi_3=\{\alpha_1, \alpha_2+\alpha_3\}\cup\{\alpha_2,
\alpha_1+\alpha_3\}\cup\{\alpha_3,
\alpha_1+\alpha_2\}\cup\{\alpha_2+\alpha_3, \alpha_1+\alpha_3,
\alpha_1-\alpha_2\}\cup\{\alpha_2+\alpha_3, \alpha_1+\alpha_2,
\alpha_1-\alpha_3\}\cup\{\alpha_1+\alpha_3, \alpha_1+\alpha_2,
\alpha_2-\alpha_3\}$. We can check that
$\alpha_1+3\alpha_2\notin\langle\bar{\alpha}, \gamma\rangle,
\forall\gamma\in\Psi_3$. Let $\beta^{(4)}=\alpha_1+3\alpha_2$ and
$\mathcal B_4=\{\alpha_1+3\alpha_2, 2\alpha_1+3\alpha_3,
2\alpha_2-\alpha_3\}$. Then the collection $\{\mathcal
B_1,\mathcal B_2,\mathcal B_3,\mathcal B_4\}$ satisfies conditions
(1)$-$(4) of Lemma \ref{genc-code}.

Similarly, we can construct a subset $\mathcal
B_5=\{2\alpha_1+3\alpha_2, \alpha_1+3\alpha_3,
\alpha_2-2\alpha_3\}$ such that the collection $\{\mathcal
B_1,\mathcal B_2,\mathcal B_3,\mathcal B_4,\mathcal B_5\}$
satisfies conditions (1)$-$(4) of Lemma \ref{genc-code}.
\end{exam}

\begin{lem}\label{lem-weak-dec-code}
Let $\{i_1,i_2\}\subseteq\{1,2,3\}$ and
$\{\Delta_1,\cdots,\Delta_K\}$ be a partition of
$\text{reg}(S_{i_1},S_{i_2})$ such that
$\text{reg}(\Delta_i)=\Delta_i, i=1,\cdots,K$. Let
$\tilde{C}_{i_1,i_2}=\{d_R;
R\in\text{reg}(S_{i_1},S_{i_2})\}\subseteq\langle\alpha_{i_1},
\alpha_{i_2}\rangle$ be such that:
\begin{itemize}
    \item[(1)] If $\{R,R'\}\subseteq\Delta_i$ for some $i\in[K]$,
    then $d_R=d_{R'}$;
    \item[(2)] If $\{R,R'\}\nsubseteq\Delta_i$ for any $i\in[K]$,
    then $d_R$ and $d_{R'}$ are linearly independent.
\end{itemize}
Then $d_R\in\langle d_{R'}; R'\in\text{In}(R)\rangle, ~\forall
R\in\text{reg}^\circ(S_{i_1},S_{i_2}).$
\end{lem}
\begin{proof}
Suppose $R\in\text{reg}^\circ(S_{i_1},S_{i_2})$. Then by
Definition \ref{g-reg},
$\text{In}(R)\subseteq\text{reg}(S_{i_1},S_{i_2})$. We have the
following two cases:

Case 1: $\text{In}(R)\subseteq\Delta_i$ for some
$i\in\{1,\cdots,K\}$. Then by Definition \ref{g-reg},
$R\in\text{reg}(\Delta_i)$. Since by the assumption of this lemma,
$\text{reg}(\Delta_i)=\Delta_i$, then $R\in\Delta_i$ and, by
condition (1), $d_R=d_{R'}$ for all $R'\in\text{In}(R)$. Thus,
$d_R\in\langle d_{R'}; R'\in\text{In}(R)\rangle$.

Case 2: $\text{In}(R)\nsubseteq\Delta_i$ for any
$i\in\{1,\cdots,K\}$. By Definition \ref{BRD}, each non-source
region has at least two parents. Since
$\{\Delta_1,\cdots,\Delta_K\}$ is a partition of
$\text{reg}(S_{i_1},S_{i_2})$, then there exists a subset
$\{i_1,i_2\}\subseteq\{1,\cdots,K\}$ such that
$\text{In}(R)\cap\Delta_{i_1}\neq\emptyset$ and
$\text{In}(R)\cap\Delta_{i_2}\neq\emptyset$. Assume
$R'_1\in\text{In}(R)\cap\Delta_{i_1}$ and
$R'_2\in\text{In}(R)\cap\Delta_{i_2}$. Then by condition (2),
$d_{R'_1}$ and $d_{R'_2}$ are linearly independent and
$d_R\in\langle d_{R'_1},d_{R'_2}\rangle=\langle
\alpha_{1},\alpha_{2}\rangle$. So $d_R\in\langle d_{R'};
R'\in\text{In}(R)\rangle$.
\end{proof}

\begin{defn}\label{ind-set-ptn}
Let $\mathcal I=\{[S_1],[S_2],[S_3],\cdots,[R_K]\}$ be a partition
of $\Pi$. A subset $\{Q,Q',Q''\}\subseteq\Pi$ is called an
$\mathcal I$-independent set if the following three conditions
hold:
\begin{itemize}
  \item [(1)] $|\{Q,Q',Q''\}\cap[[R]]|\leq 1$ for any equivalent
  class $[R]\in\mathcal I$ and any subclass $[[R]]$ of $[R]$;
  \item [(2)] $\{Q,Q',Q''\}\nsubseteq[R]$ for any equivalent
  class $R\in\mathcal I$;
  \item [(3)] $\{Q,Q',Q''\}\nsubseteq[S_{i}]_{i} \cup[S_{j}]_{j}
  \cup(\cup_{\ell=4}^K[R_\ell]_{i,j})$ for any pair
  $\{i,j\} \subseteq\{1,2,3\}$.
\end{itemize}
\end{defn}

Now we can prove Lemma \ref{comp-code}
\begin{proof}[Proof of Lemma \ref{comp-code}]
Since $\mathcal I$ is compatible, by Definition \ref{cmptl}, we
can assume $\mathcal I=\{[S_1],[S_2],[S_3],\cdots,[R_K]\}$. Let
$\mathcal B_{1}$, $\cdots$, $\mathcal B_{K}$ be as in Lemma
\ref{genc-code}. We construct a code $\tilde{C}_\Pi=\{d_R;
R\in\Pi\}\subseteq\mathbb F^3$ as follows:
\begin{itemize}
  \item For $j\in\{1,2,3\}$ and $R\in[S_j]_j$, let $d_R=\alpha_j$;
  \item For $j\in\{1,2,3\}$ and $R\in[S_j]_{i_1,i_2}$,
  let $d_R=\alpha_{i_1}+\alpha_{i_2}$, where
  $\{i_1,i_2\}=\{1,2,3\}\backslash\{j\}$;
  \item For $j\in\{4,\cdots,K\}$, $\{i_1,i_2\}\subseteq\{1,2,3\}$
  and $R\in[R_j]_{i_1,i_2}$, let $d_R=\beta^{(j)}_{i_1,i_2}$.
\end{itemize}
We shall prove that $\tilde{C}_\Pi=\{d_R;
R\in\Pi\}\subseteq\mathbb F^3$ satisfies the conditions of Lemma
\ref{lmd-solv}.

By the construction of $\tilde{C}_\Pi$, we have
$d_{S_j}=\alpha_j,j=1,2,3$. Moreover, since $\mathcal I$ is
compatible, then for each $[R_\ell]\in\mathcal I$ and
$\{i_1,i_2\}\subseteq\{1,2,3\}$, we have
$[R_\ell]_{i_1,i_2}=\text{reg}([R_\ell]_{i_1,i_2})$. $($Otherwise,
by Definition \ref{g-reg}, there is an
$R\in\text{reg}([R_\ell]_{i_1,i_2})\backslash[R_\ell]_{i_1,i_2}$.
By condition (2) of Definition \ref{cntd}, $[R_\ell]$ and $[R]$
are connected, which contradicts to the assumption that $\mathcal
I$ is compatible.$)$ Now, let $\Delta_i=[R_i]_{i_1,i_2},
i=1,\cdots,K$, where $[R_i]=[S_i], i=1,2,3$. By the construction,
$\tilde{C}_{i_1,i_2}=\{d_R; R\in\text{reg}(S_{i_1},S_{i_2})\}$
satisfies the conditions of Lemma \ref{lem-weak-dec-code}. So
$d_R\in\langle d_{R'};R'\in\text{In}(R)\rangle, \forall
R\in\text{reg}^\circ(S_{i_1},S_{i_2})$.

Finally, we prove that $\tilde{C}_\Pi$ satisfies condition (3) of
Lemma \ref{lmd-solv}. For each $\Lambda_j, j\in[n]$, we have the
following two cases:

Case 1: There is an $[R_\ell]\in\mathcal I$ such that $\Lambda_j$
intersects with at least two subclasses of $[R_\ell]$. Suppose
$Q_1\in\Lambda_j\cap[[R_\ell]]_1$ and
$Q_2\in\Lambda_j\cap[[R_\ell]]_2$, where $[[R_\ell]]_1$ and
$[[R_\ell]]_2$ are two different subclasses of $[R_\ell]$. Then by
the construction of $\tilde{C}_\Pi$, $\{d_{Q_1},
d_{Q_2}\}\subseteq\mathcal B_{\ell}$ and $\bar{\alpha}\in\langle
d_{Q_1}, d_{Q_2}\rangle$.

Case 2: For each $[R_\ell]\in\mathcal I$, $\Theta_j$ intersects
with at most one subclass of $[R_\ell]$. Since $\mathcal I$ is
compatible, then we can always find a subset
$\{Q_1,Q_2,Q_3\}\subseteq\Lambda_j$ such that $\{Q_1,Q_2,Q_3\}$ is
an $\mathcal I-$independent set. By the construction of
$\tilde{C}_\Pi$, $\{d_{Q_1}, d_{Q_2}, d_{Q_3}\}
\nsubseteq\langle\alpha_{i_1},\alpha_{i_2}\rangle,
\forall\{i_1,i_2\}\subseteq\{1,2,3\}$, and $\{d_{Q_1}, d_{Q_2},
d_{Q_3}\}\neq\{\beta^{(\ell)}_{1,2}, \beta^{(\ell)}_{1,3},
\beta^{(\ell)}_{2,3}\}, \forall \ell\in\{4,\cdots,K\}$. So
$d_{Q_1}, d_{Q_2}$ and $d_{Q_3}$ are linearly independent. Thus,
$\bar{\alpha}\in\langle d_{Q_1}, d_{Q_2}, d_{Q_3}\rangle=\mathbb
F^3$.

By the above discussion, $\tilde{C}_\Pi$ satisfies the conditions
of Lemma \ref{lmd-solv}. So $\text{RG}(D^{**})$ is feasible.
\end{proof}

Here, we make an example to illustrate the construction of
$\tilde{C}_\Pi$ in the proof of Lemma \ref{comp-code}.

\section{Proof of Lemma \ref{chpt-lem}}

Here, we prove Lemma \ref{chpt-lem}.

Suppose $\tilde{C}_\Pi=\{d_R; R\in\Pi\}\subseteq\mathbb F^3$ is a
collection that satisfies the conditions of Lemma \ref{lmd-solv}.
Note that $\text{RG}(D^{**})$ is acyclic and $d_{S_i}=\alpha_i\neq
0, i=1,2,3$. If there is an $R\in\Pi$ such that $d_R=0$, then we
can always find an $R_0\in\Pi$ such that $d_{R_0}=0$ and
$d_{R'}\neq 0, \forall R'\in\text{In}(R_0)$. We redefine $d_{R_0}$
by letting $d_{R_0}=d_{R'}$ for a fixed $R'\in\text{In}(R_0)$.
Then the resulted code $\tilde{C}_\Pi=\{d_R; R\in\Pi\}$ still
satisfies the conditions of Lemma \ref{lmd-solv} and $d_{R_0}\neq
0$. We can perform this operation continuously until $d_R\neq 0$
for all $R\in\Pi$ and the resulted code $\tilde{C}_\Pi=\{d_R;
R\in\Pi\}$ still satisfies the conditions of Lemma \ref{lmd-solv}.
So we can assume, without loss of generality, that $d_R\neq 0$ for
all $R\in\Pi$.

To prove Lemma \ref{chpt-lem}, the key is to prove that all
equivalent class $[R]\in\mathcal I_c$ satisfies the following
property:
\begin{itemize}
  \item \emph{Property (a)}: For any pair $\{Q,Q'\}\subseteq[R]$,
  $d_{Q'}\in\langle d_{Q}, \bar{\alpha}\rangle$.
\end{itemize}

To prove this, we first prove the following two lemmas.

\begin{lem}\label{claim 1}
Let $\mathcal I$ be a partition of $\Pi$ and $[R]\in\mathcal I$
satisfies Property (a). Then for any
$\{i_1,i_2\}\subseteq\{1,2,3\}$ and any pair
$\{Q,Q'\}\subseteq[R]_{i_1,i_2}$, $d_{Q'}\in\langle d_{Q}\rangle$.
Moreover, for any subclass $[[R]]$ of $[R]$ and any pair
$\{Q,Q'\}\subseteq[[R]]$, $d_{Q'}\in\langle d_{Q}\rangle$.
\end{lem}
\begin{proof}
Since $\tilde{C}_\Pi$ satisfies conditions (1) and (2) of Lemma
\ref{lmd-solv}, then by Definition \ref{g-reg}, we have
$d_W\in\langle\alpha_{i_1},\alpha_{i_2}\rangle, \forall
W\in\text{reg}(S_{i_1},S_{i_2})$. By assumption and Equation
(\ref{eq-nt-3}),
$\{Q,Q'\}\subseteq[R]_{i_1,i_2}\subseteq\text{reg}(S_{i_1},S_{i_2})$.
So $d_Q,d_{Q'}\in\langle\alpha_{i_1},\alpha_{i_2}\rangle$.
Meanwhile, since $[R]\in\mathcal I$ satisfies Property (a), then
$d_{Q'}\in\langle d_{Q}, \bar{\alpha}\rangle$. So
$d_{Q'}\in\langle\alpha_{i_1},\alpha_{i_2}\rangle\cap\langle
d_{Q},\bar{\alpha}\rangle=\langle d_{Q}\rangle$ and the first
claim is true.

We now prove the second claim. If $[R]\neq[S_i], \forall
i\in\{1,2,3\}$, then by Definition \ref{sub-class},
$[[R]]=[R]_{i_1,i_2}$ for some $\{i_1,i_2\}\subseteq\{1,2,3\}$ and
by the proven result, $d_{Q'}\in\langle d_{Q}\rangle$. If
$[R]=[S_i]$ for some $i\in\{1,2,3\}$, then by Definition
\ref{sub-class}, we have the following two cases:

Case 1: $[[R]]=[S_i]_i=[S_i]_{i,j_1}\cup[S_i]_{i,j_2}$, where
$\{j_1,j_2\}=\{1,2,3\}\backslash\{i\}$. By the proven result, we
have $\alpha_i=d_{S_i}\in\langle d_Q\rangle$ and $d_{Q'}\in\langle
d_{S_i}\rangle$. So $d_{Q'}\in\langle d_{Q}\rangle$.

Case 2: $[[R]]=[S_i]_{j_1,j_2}$, where
$\{j_1,j_2\}=\{1,2,3\}\backslash\{i\}$. By the proven result, we
have $d_{Q'}\in\langle d_{Q}\rangle$.

In both cases, we have $d_{Q'}\in\langle d_{Q}\rangle$. So the
second claim is true.
\end{proof}

\begin{lem}\label{claim 2}
Suppose $\mathcal I=\{[S_1],[S_2],[S_3],\cdots,[R_K]\}$ is a
partition of $\Pi$ in which all equivalent classes satisfy
Property (a). Suppose $\{[R],[R']\}\subseteq\mathcal I$ and there
is a $\Lambda_j$ such that $\Lambda_j\subseteq[[R']]\cup[[R'']]$,
where $[[R']]~($resp. $[[R'']])$ is a subclass of $[R']~($resp.
$[R''])$. Then $\bar{\alpha}\in\langle d_{P'},d_{P''}\rangle$ for
any $P'\in\Lambda_j\cap[[R']]$ and $P''\in\Lambda_j\cap[[R'']]$.
\end{lem}
\begin{proof}
Since $\tilde{C}_\Pi$ satisfies condition (3) of Lemma
\ref{lmd-solv} and $\Lambda_j\subseteq[[R']]\cup[[R'']]$, then
$\bar{\alpha}\in\langle d_R; R\in\Lambda_j\rangle=\langle d_R;
R\in(\Lambda_j\cap[[R']])\cup(\Lambda_j\cap[[R'']])\rangle$. By
Lemma \ref{claim 1}, $\langle d_R;
R\in(\Lambda_j\cap[[R']])\cup(\Lambda_j\cap[[R'']])\rangle=\langle
d_{P'},d_{P''}\rangle$. So $\bar{\alpha}\in\langle
d_{P'},d_{P''}\rangle$.
\end{proof}

\begin{lem}\label{claim 3}
Suppose $\mathcal I=\{[S_1],[S_2],[S_3],\cdots,[R_K]\}$ is a
partition of $\Pi$ and $\mathcal I'$ is a contraction of $\mathcal
I$ by combining two connected equivalent classes $[R']$ and
$[R'']$ in $\mathcal I$. If all equivalent classes in $\mathcal I$
satisfy Property (a), then all equivalent classes in $\mathcal I'$
satisfy Property (a).
\end{lem}
\begin{proof}
Suppose $[R]\in\mathcal I'$. If $[R]\neq[R']\cup[R'']$, then
$[R]\in\mathcal I$, and by assumption, $[R]$ satisfies property
(a). Now we suppose $[R]=[R']\cup[R'']$. Since, $[R']$ and $[R'']$
are connected, by Definition \ref{cntd}, we have the following two
cases:

Case 1: There is a $\Lambda_j\subseteq[[R']]\cup[[R'']]$, where
$[[R']]~($resp. $[[R'']])$ is a subclass of $[R']~($resp.
$[R''])$. By Lemma \ref{claim 2}, $\bar{\alpha}\in\langle
d_{P'},d_{P''}\rangle$, where $P'\in\Lambda_j\cap[[R']]$ and
$P''\in\Lambda_j\cap[[R'']]$. Then $d_{P''}\in\langle
d_{P'},\bar{\alpha}\rangle$ and $d_{P'}\in\langle
d_{P''},\bar{\alpha}\rangle$. Since, by assumption, $[R']$ and
$[R'']$ satisfy property (a), then $\langle
d_{Q},\bar{\alpha}\rangle=\langle
d_{P'},\bar{\alpha}\rangle=\langle d_{P''},\bar{\alpha}\rangle$
and $d_{Q'}\in\langle d_{Q},\bar{\alpha}\rangle, \forall
\{Q,Q'\}\subseteq[R]=[R']\cup[R'']$.

Case 2: There is a subset $\{i_1,i_2\}\subseteq\{1,2,3\}$ such
that
$\text{reg}([R']_{i_1,i_2})\cap\text{reg}([R'']_{i_1,i_2})\neq\emptyset$.
Pick a
$Q_0\in\text{reg}([R']_{i_1,i_2})\cap\text{reg}([R'']_{i_1,i_2})$.
Since, by assumption, $[R']$ and $[R'']$ satisfy property (a),
then $\langle d_{Q},\bar{\alpha}\rangle=\langle
d_{Q_0},\bar{\alpha}\rangle=\langle d_{Q'},\bar{\alpha}\rangle$
and $d_{Q'}\in\langle d_{Q},\bar{\alpha}\rangle, \forall
\{Q,Q'\}\subseteq[R]=[R']\cup[R'']$.

In both cases, $[R]=[R']\cup[R'']$ satisfies property (a). Thus,
all equivalent classes in $\mathcal I'$ satisfy Property (a).
\end{proof}

Now we can prove Lemma \ref{chpt-lem}.
\begin{proof}[Proof of Lemma \ref{chpt-lem}]
Since each equivalent class $[R]$ in $\mathcal I_0$ contains
exactly one region $R$, so $[R]$ naturally satisfies property (a)
and $[S_i]\neq[S_j]$ for all $\{i,j\}\subseteq\{1,2,3\}$.

By Definition \ref{chpt}, $\mathcal I_c=\mathcal I_L$, where
$\mathcal I_0,\mathcal I_1,\cdots,\mathcal I_L=\mathcal I_c$ is a
sequence of partitions of $\Pi$ such that $\mathcal I_\ell$ is a
contraction of $\mathcal I_{\ell-1}$ by combining two connected
equivalent classes in $\mathcal I_{\ell-1}$ and, for any
$\{i,j\}\subseteq\{1,2,3\}$, $[S_i]\neq[S_j]$ in $\mathcal
I_{\ell-1}$, $\ell=1,\cdots,L$. So by Lemma \ref{claim 3}, all
equivalent classes in $\mathcal I_\ell$ satisfy Property (a). In
particular, all equivalent classes in $\mathcal I_c=\mathcal I_L$
satisfies property (a). Then the conclusion of Lemma
\ref{chpt-lem} is obtained by Lemma \ref{claim 1}.
\end{proof}

\section{Proof of Theorem \ref{3s-3t}}

Here, we prove Theorem \ref{3s-3t}. First, we prove some lemmas.

\begin{lem}\label{3s-2t}
If $\text{RG}(D^{**})$ has two terminal regions, then
$\text{RG}(D^{**})$ is feasible.
\end{lem}
\begin{proof}
Suppose $\text{RG}(D^{**})$ has two terminal regions $T_1$ and
$T_2$. We have the following two cases:

Case 1: $\Omega_{1,2}\neq\emptyset$. Then there is a $Q\in
D^{**}\backslash\Pi$ such that $Q\rightarrow T_i, i=1,2$. Similar
to what we did in Remark \ref{non-ts}, we can first construct a
code on the set $\{R\in D^{**}; R\rightarrow Q\}$ such that
$d_Q=\bar{\alpha}$. Then for all $R$ such that $Q\rightarrow
R\rightarrow T_i$ for some $i\in\{1,2\}$, let $d_R=\bar{\alpha}$.
By this construction, we obtain a solution of $\text{RG}(D^{**})$.
So $\text{RG}(D^{**})$ is feasible.

Case 2: $\Omega_{1,2}=\emptyset$. Then $\text{RG}(D^{**})$ is
terminal separable. From Lemma \ref{in-reg-lmd}, we have
$\Lambda_j\nsubseteq\text{reg}(S_{i_1},S_{i_2}),
\forall\{i_1,i_2\}\subseteq\{1,2,3\},$ and $|\Lambda_j|\geq 2,
j=1,2$. By enumerating, we have the following three subcases:

Case 2.1: $|\Lambda_1|>2$ and $|\Lambda_2|>2$. Let $\mathcal
I_0=\{[R];R\in\Pi\}$, where $[R]=\{R\}$ for all $R\in\Pi$.
Clearly, $\mathcal I_0$ is a partition of $\Pi$ and is compatible.
By Lemma \ref{comp-code}, $\text{RG}(D^{**})$ is feasible.

Case 2.2: $|\Lambda_1|>2$ and $|\Lambda_2|=2$. Let $\mathcal
I=\{\Lambda_{2}\}\cup\{[R];R\in\Pi\backslash\Lambda_{2}\}$, where
$[R]=\{R\}$ for all $R\in\Pi\backslash\Lambda_{2}$. Clearly,
$\mathcal I$ is a partition of $\Pi$ and is compatible. By Lemma
\ref{comp-code}, $\text{RG}(D^{**})$ is feasible.

Case 2.3: $|\Lambda_1|=|\Lambda_2|=2$ and
$\Lambda_1\cap\Lambda_2=\emptyset$. Let $\mathcal
I=\{\Lambda_{1},\Lambda_{2}\}\cup\{[R];R\in\Pi\backslash
(\Lambda_{1}\cup\Lambda_{2})\}$, where $[R]=\{R\}$ for all
$R\in\Pi\backslash(\Lambda_{1}\cup\Lambda_{2})$. Clearly,
$\mathcal I$ is a partition of $\Pi$ and is compatible. By Lemma
\ref{comp-code}, $\text{RG}(D^{**})$ is feasible.

Case 2.4: $|\Lambda_1|=|\Lambda_2|=2$ and
$\Lambda_1\cap\Lambda_2\neq\emptyset$. If $\Lambda_1=\Lambda_2$,
then it is easy to construct a code $\tilde{C}_\Pi$ satisfies the
conditions of Lemma \ref{lmd-solv}. So $\text{RG}(D^{**})$ is
feasible. Thus, we can assume $\Lambda_1\neq\Lambda_2$. Then by
proper naming, we can assume $\Lambda_1=\{Q_1,Q_2\}$ and
$\Lambda_2=\{Q_1,Q_3\}$. By Lemma \ref{in-reg-lmd},
$\{Q_1,Q_2\}\nsubseteq\text{reg}(S_{i_1},S_{i_2})$ and
$\{Q_1,Q_3\}\nsubseteq\text{reg}(S_{i_1},S_{i_2})$ for all
$\{i_1,i_2\}\subseteq\{1,2,3\}$. Then one of the following two
cases hold:

Case 2.4.1: $\{Q_2,Q_3\}\subseteq\text{reg}(S_{i_1},S_{i_2})$ for
some $\{i_1,i_2\}\subseteq\{1,2,3\}$. Let $\mathcal
I=\{[Q_1]\}\cup\{[R];R\in\Pi\backslash[Q_1]\}$, where
$[Q_1]=\{Q_1\}\cup\text{reg}(Q_2,Q_3)$ and $[R]=\{R\}$ for all
$R\in\Pi\backslash[Q_1]$. Clearly, $\mathcal I$ is a partition of
$\Pi$ and is compatible. By Lemma \ref{comp-code},
$\text{RG}(D^{**})$ is feasible.

Case 2.4.2: $\{Q_2,Q_3\}\nsubseteq\text{reg}(S_{i_1},S_{i_2})$ for
all $\{i_1,i_2\}\subseteq\{1,2,3\}$. Let $\mathcal
I=\{[Q_1]\}\cup\{[R];R\in\Pi\backslash[Q_1]\}$, where
$[Q_1]=\{Q_1,Q_2,Q_3\}$ and $[R]=\{R\}$ for all
$R\in\Pi\backslash[Q_1]$. Clearly, $\mathcal I$ is a partition of
$\Pi$ and is compatible. By Lemma \ref{comp-code},
$\text{RG}(D^{**})$ is feasible.

By the above discussions, we proved that $\text{RG}(D^{**})$ is
feasible.
\end{proof}


\begin{lem}\label{simple-case}
Suppose $\text{RG}(D^{**})$ has three terminal regions and is
terminal separable. Then $\text{RG}(D^{**})$ is feasible if one of
the following conditions hold:
\begin{itemize}
  \item [(1)] $|\Lambda_{j_1}|\geq 3$ and $|\Lambda_{j_2}|\geq 3$
  for some $\{j_1,j_2\}\subseteq\{1,2,3\}$;
  \item [(2)] For any $\{j_1,j_2\}\subseteq\{1,2,3\}$,
  if $|\Lambda_{j_1}|=|\Lambda_{j_2}|=2$, then
  $\Lambda_{j_1}\cap\Lambda_{j_2}=\emptyset$;
  \item [(3)] $S_i\notin\Lambda_j$ for all $i,j\in\{1,2,3\}$;
  \item [(4)] There is a subset $\{\ell',\ell''\}\subseteq\{1,2,3\}$
  such that $\Lambda_j\cap\text{reg}^\circ(S_{\ell'},S_{\ell''})
  \neq\emptyset$ for all $j\in\{1,2,3\}$.
\end{itemize}
\end{lem}
\begin{proof}
1) Suppose condition (1) holds. Let
$j_3\in\{1,2,3\}\backslash\{j_1,j_2\}$. From Lemma
\ref{in-reg-lmd}, we have
$\Lambda_{j_3}\nsubseteq\text{reg}(S_{i_1},S_{i_2}),
\forall\{i_1,i_2\}\subseteq\{1,2,3\},$ and $|\Lambda_{j_3}|\geq
2$. Then we have the following two cases:

Case 1: $|\Lambda_{j_3}|>2$. Let $\mathcal I_0=\{[R];R\in\Pi\}$,
where $[R]=\{R\}$ for all $R\in\Pi$. Clearly, $\mathcal I_0$ is a
partition of $\Pi$ and is compatible. By Lemma \ref{comp-code},
$\text{RG}(D^{**})$ is feasible.

Case 2: $|\Lambda_{j_3}|=2$. Let $\mathcal
I=\{\Lambda_{j_3}\}\cup\{[R];R\in\Pi\backslash\Lambda_{j_3}\}$,
where $[R]=\{R\}$ for all $R\in\Pi\backslash\Lambda_{j_3}$.
Clearly, $\mathcal I$ is a partition of $\Pi$ and is compatible.
By Lemma \ref{comp-code}, $\text{RG}(D^{**})$ is feasible.

2) Suppose condition (2) holds. Let $A\subseteq\{1,2,3\}$ be such
that $|\Lambda_j|=2, \forall j\in A$, and $|\Lambda_j|>2, \forall
j\in\{1,2,3\}\backslash A$. Let $\mathcal I=\{\Lambda_{j};j\in
A\}\cup\{[R];R\in\Pi\backslash(\cup_{j\in A}\Lambda_{j})\}$, where
$[R]=\{R\}$ for all $R\in\Pi\backslash(\cup_{j\in A}\Lambda_{j})$.
Clearly, $\mathcal I$ is a partition of $\Pi$ and is compatible.
By Lemma \ref{comp-code}, $\text{RG}(D^{**})$ is feasible.

3) Suppose condition (3) holds. Let $\mathcal
I=\{[S_1],[S_2],[S_3]$, $[R]\}$, where $[S_i]=\{S_i\}$ for
$i\in\{1,2,3\}$ and $[R]=\Pi\backslash\{S_1,S_2,S_3\}$. Clearly,
$\mathcal I$ is a partition of $\Pi$ and is compatible. By Lemma
\ref{comp-code}, $\text{RG}(D^{**})$ is feasible.

4) Without loss of generality, assume $\ell=1, \ell'=2$ and
$\ell''=3$. Let $\mathcal I=\{[S_1],[S_2],[S_3]\}$, where
$[S_1]=\text{reg}(S_1,S_2)
\cup\text{reg}(S_1,S_3)\cup\text{reg}^\circ(S_2,S_3)$,
$[S_2]=\{S_2\}$ and $[S_3]=\{S_3\}$. Clearly, $\mathcal I$ is a
partition of $\Pi$ and is compatible. By Lemma \ref{comp-code},
$\text{RG}(D^{**})$ is feasible.
\end{proof}

\begin{lem}\label{sgl-3t}
Suppose $\text{RG}(D^{**})$ has three terminal regions and is
terminal separable. If $\text{RG}(D^{**})$ is not feasible, then
the condition (C-IR) holds.
\end{lem}
\begin{proof}
If $\Lambda_{j_1}=\Lambda_{j_2}$ for some
$\{j_1,j_2\}\subseteq\{1,2,3\}$, then by lemma \ref{3s-2t}, we can
construct a code $\tilde{C}_{\Pi}$ satisfies the conditions of
Lemma \ref{lmd-solv}. So $\text{RG}(D^{**})$ is feasible. Thus, we
assume $\Lambda_1,\Lambda_2$ and $\Lambda_3$ are mutually
different. Since $\text{RG}(D^{**})$ is not feasible, then by (1),
(2) of Lemma \ref{simple-case}, there is a
$\{j_1,j_2\}\subseteq\{1,2,3\}$ such that
\begin{align}
|\Lambda_{j_1}|=|\Lambda_{j_2}|=2 ~\text{and}~
|\Lambda_{j_1}\cap\Lambda_{j_2}|=1.\label{eq-sgl-3t-1}
\end{align}
Let $j_3\in\{1,2,3\}\backslash\{j_1,j_2\}$. By enumerating, we can
divide our discussion into the following cases:

Case 1:
$\Lambda_{j_1}\cup\Lambda_{j_2}\subseteq\text{reg}^\circ(S_1,S_2)\cup
\text{reg}^\circ(S_1,S_3)\cup\text{reg}^\circ(S_2,S_3).$ By
(\ref{eq-sgl-3t-1}), we can assume
\begin{align}
\Lambda_{j_1}=\{P_1,P_2\} ~\text{and}~
\Lambda_{j_2}=\{P_0,P_2\}.\label{eq-sgl-3t-4}
\end{align}
By Lemma \ref{in-reg-lmd},
$\Lambda_{j_1}\nsubseteq\text{reg}(S_{i_1},S_{i_2}),
\forall\{i_1,i_2\}\subseteq\{1,2,3\}$. Then by proper naming, we
can assume
\begin{align}
P_2\in\text{reg}^\circ(S_{\ell_1},S_{\ell_2}) ~\text{and}~
P_1\in\text{reg}^\circ(S_{\ell_2},S_{\ell_3}).\label{eq-sgl-3t-5}
\end{align}
where $\{\ell_1,\ell_2,\ell_3\}$ is a fixed permutation of
$\{1,2,3\}$. Also, by Lemma \ref{in-reg-lmd},
$\Lambda_{j_2}\nsubseteq\text{reg}(S_{i_1},S_{i_2}),
\forall\{i_1,i_2\}\subseteq\{1,2,3\}$. Then for $P_0$, we have the
following subcases:

Case 1.1: $P_0\in\text{reg}^\circ(S_{\ell_2},S_{\ell_3})$. We can
further divide this case into the following two subcases:

Case 1.1.1:
$\Lambda_{j_3}\cap(\text{reg}^\circ(S_{\ell_1},S_{\ell_2})
\cup\text{reg}^\circ(S_{\ell_2},S_{\ell_3}))\neq\emptyset$. By (4)
of Lemma \ref{simple-case}, $\text{RG}(D^{**})$ is feasible.

Case 1.1.2:
$\Lambda_{j_3}\cap(\text{reg}^\circ(S_{\ell_1},S_{\ell_2})
\cup\text{reg}^\circ(S_{\ell_2},S_{\ell_3}))=\emptyset$. Then
$\Lambda_{j_3}\subseteq\text{reg}(S_{\ell_1},S_{\ell_3})\cup\{S_{\ell_2}\}$.
Moreover, since by Lemma \ref{in-reg-lmd},
$\Lambda_{j_3}\nsubseteq\text{reg}(S_{i_1},S_{i_2}),\forall
\{i_1,i_2\}\subseteq\{1,2,3\}$, then either
$\Lambda_{j_3}=\{S_1,S_2,S_3\}$ or
$\{Q,S_{\ell_2}\}\subseteq\Lambda_{j_3}$ for some
$Q\in\text{reg}^\circ(S_{\ell_1},S_{\ell_3})$.

Let $\mathcal I=\{[S_{\ell_1}],[S_{\ell_2}],[S_{\ell_3}],[P_2]\}$,
where $[S_{\ell_1}]=\{S_{\ell_1}\}$,
$[S_{\ell_2}]=\{S_{\ell_2}\}\cup\text{reg}^\circ(S_{\ell_1},S_{\ell_3})$,
$[S_{\ell_3}]=\{S_{\ell_3}\}$ and
$[P_2]=\text{reg}^\circ(S_{\ell_1},
S_{\ell_2})\cup\text{reg}^\circ(S_{\ell_2},S_{\ell_3})$. Clearly,
$\mathcal I$ is a partition of $\Pi$ and is compatible. By Lemma
\ref{comp-code}, $\text{RG}(D^{**})$ is feasible.

Case 1.2: $P_0\in\text{reg}^\circ(S_{\ell_1},S_{\ell_3})$. This
case can be further divided into the following subcases:

Case 1.2.1: $|\Lambda_{j_3}|=3$ or
$\Lambda_{j_3}\subseteq\{P_0,P_1,P_2\}$. Let $\mathcal
I=\{[P_2]\}\cup\{[R]; R\in\Pi\backslash[P_0]\}$, where
$[P_2]=\{P_0,P_1,P_2\}$ and $[R]=\{R\}, \forall
R\in\Pi\backslash[P_2]$. Clearly, $\mathcal I$ is a partition of
$\Pi$ and is compatible. By Lemma \ref{comp-code},
$\text{RG}(D^{**})$ is feasible.

Case 1.2.2: $|\Lambda_{j_3}|=2$ and
$\Lambda_{j_3}\cap\{P_0,P_1,P_2\}=\emptyset$. Assume
$\Lambda_{j_3}=\{P_3,P_4\}$. Let $\mathcal
I=\{[P_2],[P_3]\}\cup\{[R]; R\in\Pi\backslash([P_2]\cup[P_3])\}$,
where $[P_2]=\{P_0,P_1,P_2\}, [P_3]=\{P_3,P_4\}$ and $[R]=\{R\},
\forall R\in\Pi\backslash([P_2]\cup[P_3])$. Clearly, $\mathcal I$
is a partition of $\Pi$ and is compatible. By Lemma
\ref{comp-code}, $\text{RG}(D^{**})$ is feasible.

Case 1.2.3: $|\Lambda_{j_3}|=2$ and
$\Lambda_{j_3}\cap\{P_0,P_1,P_2\}=\{P_2\}$. By (4) of Lemma
\ref{simple-case}, $\text{RG}(D^{**})$ is feasible.

Case 1.2.4: $|\Lambda_{j_3}|=2$ and
$\Lambda_{j_3}\cap\{P_0,P_1,P_2\}=\{P_1\}~($or $\{P_0\})$. By
proper naming, we can assume $\Lambda_{j_3}=\{P_1,P_3\}$, where
$P_3\notin\{P_0,P_1,P_2\}$. If $P_3\neq S_\ell,
\forall\ell\in\{1,2,3\}$, then by (3) of Lemma \ref{simple-case},
$\text{RG}(D^{**})$ is feasible. So we assume $P_3=S_\ell$ for
some $\ell\in\{1,2,3\}$. Since
$P_1\in\text{reg}^\circ(S_{\ell_2},S_{\ell_3})$ and, by Lemma
\ref{in-reg-lmd},
$\Lambda_{j_3}=\{P_1,P_3\}\nsubseteq\text{reg}(S_{i_1},S_{i_2}),
\forall \{i_1,i_2\}\subseteq\{1,2,3\}$, then $P_3=S_{\ell_1}$. Let
$j_3=1, j_1=2, j_2=3$ and $\ell_i=i~ (i=1,2,3)$. Then the
condition (C-IR) holds.

Case 2: There is an $\ell_1\in\{1,2,3\}$ such that
$S_{\ell_1}\in\Lambda_{j_1}\cup\Lambda_{j_2}$. Let
$\{\ell_2,\ell_3\}=\{1,2,3\}\backslash\{\ell_1\}$. We can further
divide this case into the following subcases:

Case 2.1: $S_{\ell_1}\in\Lambda_{j_1}\cap\Lambda_{j_2}$. By proper
naming, we assume
\begin{align}
\Lambda_{j_1}=\{S_{\ell_1},P_1\} ~\text{and}~
\Lambda_{j_2}=\{S_{\ell_1},P_2\}.\label{eq-sgl-3t-6}
\end{align}
Since, by Lemma \ref{in-reg-lmd}, $\Lambda_{j_1}$,
$\Lambda_{j_2}\nsubseteq\text{reg}(S_{i_1},S_{i_2}),
\forall\{i_1,i_2\}\subseteq\{1,2,3\}$, then we have
\begin{align}
P_1,P_2\in\text{reg}^\circ(S_{\ell_2},S_{\ell_3}).\label{eq-sgl-3t-7}
\end{align}
If
$\Lambda_{j_3}\cap\text{reg}^\circ(S_{\ell_2},S_{\ell_3})\neq\emptyset$,
then by (4) of Lemma \ref{simple-case}, $\text{RG}(D^{**})$ is
feasible. So we assume
$\Lambda_{j_3}\cap\text{reg}^\circ(S_{\ell_2},S_{\ell_3})=\emptyset$.
Then
\begin{align}
\Lambda_{j_3}\subseteq\text{reg}(S_{\ell_1},S_{\ell_2})\cup
\text{reg}(S_{\ell_1},S_{\ell_3}). \label{eq-sgl-3t-8}
\end{align}
We have the following two subcases:

Case 2.1.1:
$\Lambda_{j_3}\cap(\text{reg}^\circ(S_{\ell_1},S_{\ell_2})\cup
\text{reg}^\circ(S_{\ell_1},S_{\ell_3}))\neq\emptyset$. Without
loss of generality, assume
$Q_1\in\Lambda_{j_3}\cap\text{reg}^\circ(S_{\ell_1},S_{\ell_2})$.
Since, by Lemma \ref{in-reg-lmd},
$\Lambda_{j_3}\nsubseteq\text{reg}(S_{i_1},S_{i_2}),\forall
\{i_1,i_2\}\subseteq\{1,2,3\}$, then there is a
$Q_2\in\text{reg}(S_{\ell_1},S_{\ell_3})\backslash\{S_{\ell_1}\}$
such that $Q_2\in\Lambda_{j_3}$. Let $\mathcal I=\{[S_{\ell_1}],
[S_{\ell_3}]\}\cup\{[R];
R\in\Pi\backslash([S_{\ell_1}]\cup[S_{\ell_3}])\}$, where
$[S_{\ell_1}]=\{S_{\ell_1}\}\cup\text{reg}(P_1,P_2)$,
$[S_{\ell_3}]=\{Q_1\}\cup\text{reg}(S_{\ell_1},
S_{\ell_3})\backslash\{S_{\ell_1}\}$ and $[R]=\{R\}, \forall
R\in\Pi\backslash([S_{\ell_1}]\cup[S_{\ell_3}])$. Clearly,
$\mathcal I$ is a partition of $\Pi$ and is compatible. By Lemma
\ref{comp-code}, $\text{RG}(D^{**})$ is feasible.

Case 2.1.2:
$\Lambda_{j_3}\cap(\text{reg}^\circ(S_{\ell_1},S_{\ell_2})\cup
\text{reg}^\circ(S_{\ell_1},S_{\ell_3}))=\emptyset$. By
(\ref{eq-sgl-3t-8}), $\Lambda_{j_3}\subseteq\{S_1,S_2,S_3\}$.
Since, by Lemma \ref{in-reg-lmd},
$\Lambda_{j_3}\nsubseteq\text{reg}(S_{i_1},S_{i_2}),
\forall\{i_1,i_2\}\subseteq\{1,2,3\}$, then
$\Lambda_{j_3}=\{S_1,S_2,S_3\}$. Let $\mathcal
I=\{[S_{\ell_1}]\}\cup\{[R]; R\in\Pi\backslash[S_{\ell_1}]\}$,
where $[S_{\ell_1}]=\{S_{\ell_1}\}\cup\text{reg}(P_1,P_2)$ and
$[R]=\{R\}, \forall R\in\Pi\backslash[S_{\ell_1}]$. Clearly,
$\mathcal I$ is a partition of $\Pi$ and is compatible. By Lemma
\ref{comp-code}, $\text{RG}(D^{**})$ is feasible.

Case 2.2: $S_{\ell_1}\notin\Lambda_{j_1}\cap\Lambda_{j_2}$. Since
$S_{\ell_1}\in\Lambda_{j_1}\cup\Lambda_{j_2}$, then by
(\ref{eq-sgl-3t-1}) and proper naming, we can assume
$\Lambda_{j_1}=\{S_{\ell_1},P_1\},\Lambda_{j_2}=\{P_{1},P_2\}$.
Since, by Lemma \ref{in-reg-lmd}, $\Lambda_{j_1},
\Lambda_{j_2}\nsubseteq\text{reg}(S_{i_1},S_{i_2}),
\forall\{i_1,i_2\}\subseteq\{1,2,3\}$, then
$P_1\in\text{reg}^\circ(S_{\ell_2},S_{\ell_3})$ and, by proper
naming, we can assume
$P_2\in\text{reg}^\circ(S_{\ell_1},S_{\ell_2})$, where
$\{\ell_2,\ell_3\}=\{1,2,3\}\backslash\{\ell_1\}$. Let
$j_3\in\{1,2,3\}\backslash\{j_1,j_2\}$. If
$\Lambda_{j_3}\cap\text{reg}^\circ(S_{\ell_2},S_{\ell_3})\neq\emptyset$,
then by (4) of Lemma \ref{simple-case}, $\text{RG}(D^{**})$ is
feasible. So we assume
$\Lambda_{j_3}\cap\text{reg}^\circ(S_{\ell_2},S_{\ell_3})=\emptyset$.
Then
\begin{align}
\Lambda_{j_3}\subseteq\text{reg}(S_{\ell_1},S_{\ell_2})\cup
\text{reg}(S_{\ell_1},S_{\ell_3}). \label{eq-sgl-3t-2}
\end{align}
Now, suppose
\begin{align}
\Lambda_{j_3}\nsubseteq\text{reg}(S_{\ell_1},P_2)\cup\text{reg}(S_{\ell_1},
S_{\ell_3}). \label{eq-sgl-3t-3}
\end{align}
We shall prove $\text{RG}(D^{**})$ is feasible. We have the
following three subcases:

Case 2.2.1: $\Lambda_{j_3}\cap\text{reg}(S_{\ell_1},P_2)\neq
\emptyset$. Since, by Lemma \ref{in-reg-lmd},
$\Lambda_{j_3}\nsubseteq\text{reg}(S_{i_1},S_{i_2})$ for all
$\{i_1,i_2\}\subseteq\{1,2,3\}$, then by (\ref{eq-sgl-3t-2}),
$\Lambda_{j_3}\cap(\text{reg}(S_{\ell_1},S_{\ell_3})
\backslash\{S_{\ell_1}\})\neq\emptyset.$ Moreover, by
(\ref{eq-sgl-3t-3}),
$\Lambda_{j_3}\cap(\text{reg}(S_{\ell_1},S_{\ell_2})
\backslash\text{reg}(S_{\ell_1},P_2))\neq\emptyset$. Let $\mathcal
I=\{[S_{\ell_1}]\}\cup\{[R]; R\in\Pi\backslash[S_{\ell_1}]\}$,
where $[S_{\ell_1}]=\text{reg}(S_{\ell_1},P_2)\cup\{P_1\}$ and
$[R]=\{R\}, \forall R\in\Pi\backslash[S_{\ell_1}]\}$. Then
$\mathcal I$ is a partition of $\Pi$ and is compatible. By Lemma
\ref{comp-code}, $\text{RG}(D^{**})$ is feasible.

Case 2.2.2:
$\Lambda_{j_3}\cap\text{reg}(S_{\ell_1},P_2)=\emptyset$ and
$|\Lambda_{j_3}|\geq 3.$ As in Case 2.2.1, we can prove $\mathcal
S$ is regular.

Case 2.2.3:
$\Lambda_{j_3}\cap\text{reg}(S_{\ell_1},P_2)=\emptyset$ and
$|\Lambda_{j_3}|=2.$ Since, by Lemma \ref{in-reg-lmd},
$\Lambda_{j_3}\nsubseteq\text{reg}(S_{i_1},S_{i_2}),\forall
\{i_1,i_2\}\subseteq\{1,2,3\}$. Then by Equations
(\ref{eq-sgl-3t-2}), (\ref{eq-sgl-3t-3}) and proper naming, we can
assume $\Lambda_{j_3}=\{P_3,P_4\}$, where
$P_3\in\text{reg}(S_{\ell_1},S_{\ell_2})
\backslash(\text{reg}(S_{\ell_1},P_2)\cup\{S_{\ell_2}\})$ and
$P_4\in\text{reg}(S_{\ell_1},S_{\ell_3})
\backslash\{S_{\ell_1}\}$. Let $\mathcal
I=\{[S_{\ell_1}],[P_3]\}\cup\{[R];
R\in\Pi\backslash[S_{\ell_1}]\cup[P_3]\}$, where
$[S_{\ell_1}]=\text{reg}(S_{\ell_1},P_2)\cup\{P_1\}$,
$[P_3]=\{P_3,P_4\}$ and $[R]=\{R\}, \forall
R\in\Pi\backslash[S_{\ell_1}]\cup[P_3]\}$. Then $\mathcal I$ is a
partition of $\Pi$ and is compatible. By Lemma \ref{comp-code},
$\text{RG}(D^{**})$ is feasible.

So for case 2.2, if $\text{RG}(D^{**})$ is not feasible, then
$\Lambda_{j_3}\subseteq\text{reg}(S_{\ell_1},P_2)
\cup\text{reg}(S_{\ell_1}, S_{\ell_3})$. Let
$\ell_i=j_i=i,i=1,2,3$. Then the condition (C-IR) holds.

Combining the discussions for all cases above, we can conclude
that if $\text{RG}(D^{**})$ is not feasible, then the condition
(C-IR) holds.
\end{proof}

\begin{lem}\label{sgl-3t-infsb}
Suppose $\text{RG}(D^{**})$ has three terminal regions and is
terminal separable. If the condition (C-IR) holds, then
$\text{RG}(D^{**})$ is not feasible.
\end{lem}
\begin{proof}
We prove this lemma by contradiction. For this purpose, we suppose
$\text{RG}(D^{**})$ is feasible and the condition (C-IR) holds.
Then there is a code
$\tilde{C}_\Pi=\{d_R;R\in\Pi\}\subseteq\mathbb F^3$ satisfying
conditions of Lemma \ref{lmd-solv}. Since
$P_1\in\text{reg}^\circ(S_2,S_3)$, then by Definition \ref{g-reg}
and condition (2) of Lemma \ref{lmd-solv}, we have
$$d_{P_1}\in\langle\alpha_{2},\alpha_{3}\rangle.$$ Moreover,
since $\Lambda_1=\{S_1,P_1\}$, then by conditions (1), (3) of
Lemma \ref{lmd-solv}, we have
$\bar{\alpha}\in\langle\alpha_1,d_{P_1}\rangle$. So
$$d_{P_1}\in\langle\alpha_1,\bar{\alpha}\rangle\cap\langle\alpha_{2},
\alpha_{3}\rangle=\langle\alpha_{2}+\alpha_{3}\rangle.$$
Similarly, since $\Lambda_2=\{P_1,P_2\}$ and
$P_2\in\text{reg}^\circ(S_1,S_2)$, then
$$d_{P_2}\in\langle d_{P_1},\bar{\alpha}\rangle\cap\langle\alpha_{1},
\alpha_{2}\rangle=\langle\alpha_{1}\rangle.$$ By Definition
\ref{g-reg} and condition (2) of Lemma \ref{lmd-solv},
$d_R\in\langle\alpha_{1}\rangle$ for all
$R\in\text{reg}(S_{1},P_2)$ and $d_R\in\langle\alpha_{1},
\alpha_{3}\rangle$ for all $R\in\text{reg}(S_{1},S_{3})$. Since
$\Lambda_3\subseteq\text{reg}(S_1,P_2)\cup\text{reg}(S_1,S_3)$,
then $$\langle d_R;
R\in\Lambda_3\rangle\subseteq\langle\alpha_{1},
\alpha_{3}\rangle.$$ By condition (3) of Lemma \ref{lmd-solv}, we
have $\bar{\alpha}\in\langle
d_R;R\in\Lambda_3\rangle\subseteq\langle\alpha_{1},
\alpha_{3}\rangle$, a contradiction. Thus, we can conclude that if
the condition (C-IR) holds, then $\text{RG}(D^{**})$ is not
feasible.
\end{proof}

Now, we can prove Theorem \ref{3s-3t}.
\begin{proof}[Proof of Theorem \ref{3s-3t}]
By enumerating, one of the following three cases hold:

Case 1: $\Omega_{1,2,3}\neq\emptyset$. Then there is a $Q\in
D^{**}\backslash\Pi$ such that $Q\rightarrow T_i, i=1,2,3$.
Similar to what we did in Remark \ref{non-ts}, we can first
construct a code on the set $\{R\in D^{**}; R\rightarrow Q\}$ such
that $d_Q=\bar{\alpha}$. Then for all $R$ such that $Q\rightarrow
R\rightarrow T_i$ for some $i\in\{1,2,3\}$, let
$d_R=\bar{\alpha}$. By this construction, we obtain a solution of
$\text{RG}(D^{**})$. So $\text{RG}(D^{**})$ is feasible.

Case 2: $\Omega_{1,2,3}=\emptyset$ and
$\Omega_{i_1,i_2}\neq\emptyset$ for some
$\{i_1,i_2\}\subseteq\{1,2,3\}$. Then there is a $Q\in
D^{**}\backslash\Pi$ such that $Q\rightarrow T_{i_1}$ and
$Q\rightarrow T_{i_2}$. Let
$i_3\in\{1,2,3\}\backslash\{i_1,i_2\}$. We can view $T_{i_3}$ and
$Q$ as two terminal regions and, by Lemma \ref{3s-2t}, we can
construct a code on the set $\{R\in D^{**}; R\rightarrow Q~
\text{or}~ R\rightarrow T_{i_3}\}$ such that
$d_Q=d_{T_{i_3}}=\bar{\alpha}$. Moreover, for all $R$ such that
$Q\rightarrow R\rightarrow T_{i_1}$ or $Q\rightarrow R\rightarrow
T_{i_2}$, let $d_R=\bar{\alpha}$. Then we obtain a solution of
$\text{RG}(D^{**})$. So $\text{RG}(D^{**})$ is feasible.

Case 3: $\text{RG}(D^{**})$ is terminal separable. By Lemma
\ref{sgl-3t} and \ref{sgl-3t-infsb}, $\text{RG}(D^{**})$ is not
feasible if and only if, by proper naming, the condition (C-IR)
holds.

By the above discussion, we proved Theorem \ref{3s-3t}.
\end{proof}

\begin{thebibliography}{1}
\bibitem{Ahlswede00}
R. Ahlswede, N. Cai, S.-Y. R. Li, and R. W. Yeung, ``Network
information flow,'' \emph{IEEE Trans. Inf. Theory}, vol. 46, no.
4, pp. 1204-1216, Jul. 2000.

\bibitem{Li03}
S.-Y. R. Li, R. W. Yeung, and N. Cai, ``Linear network coding,''
\emph{IEEE Trans. Inf. Theory}, vol. 49, no. 2, pp. 371-381,
Feb.2003.

\bibitem{Rama08}
A. Ramamoorthy, ``Communicating the sum of sources over a
network,'' in \emph{Proc ISIT}, Toronto, Canada, July 06-11, pp.
1646-1650, 2008.

\bibitem{Rai091}
B. K. Rai, B. K. Dey, and A. Karandikar, ``Some results on
communicating the sum of sources over a network,'' in \emph{Proc
NetCod} 2009.



\bibitem{Langberg09}
M. Langberg and A. Ramamoorthy, ``Communicating the sum of sources
in a 3-sources/3-terminals network,'' in \emph{Proc ISIT}, Seoul,
Korea, 2009.



\bibitem{Shenvi10}
S. Shenvi and B. K. Dey, ``A necessary and sufficient condition
for solvability of a 3s/3t sum network,''  in \emph{Proc ISIT},
Texas, U.S.A, 2010.

\bibitem{Rai12}
B. K. Rai and B. K. Dey, ``On Network Coding for Sum-Networks,''
\emph{IEEE Trans. Inf. Theory}, vol. 58, no. 1, pp. 50-63, Jan.
2012.

\bibitem{Rai13}
B. K. Rai and N. Das, ``Sum-Networks: Min-Cut=2 Does Not Guarantee
Solvability,'' \emph{IEEE Communications Letters}, vol. 17, no.
11, pp. 2144-2147, Nov. 2013.


\bibitem{Giridhar}
A. Giridhar and P. R. Kumar, ``Computing and communicating
functions over sensor networks,'' \emph{IEEE J. Select. Areas
Commun.}, vol. 23, no. 4, pp. 755-764, 2005.

\bibitem{Kanoria}
Y. Kanoria and D. Manjunath, ``On distributed computation in noisy
random planar networks,'' in \emph{Proc. ISIT}, Nice, France,
2008.

\bibitem{Appuswamy}
R. Appuswamy, M. Franceschetti, N. Karamchandani, and K. Zeger,
¡°Network coding for computing: cut-set bounds,¡± \emph{IEEE
Trans. Inf. Theory}, vol. 57, no. 2, pp. 1015-1030, Feb. 2011.

\bibitem{Kannan}
S. Kannan and P. Viswanath, ¡°Multi-session function computation
and multicasting in undirected graphs,¡± \emph{IEEE J. Select.
Areas Commun.}, vol. 31, no. 4, pp. 702-713, 2013.

\bibitem{Fragouli06}
C. Fragouli and E. Soljanin,``Information flow decomposition for
network coding,'' \emph{IEEE Trans. Inf. Theory}, vol. 52, no. 3,
pp. 829-848, Mar. 2006.

\bibitem{Wentu11}
W. Song, K. Cai, R. Feng and C. Yuen, ``The Complexity of Network
Coding With Two Unit-Rate Multicast Sessions,'' \emph{IEEE Trans.
Inf. Theory}, vol. 59, no. 9, pp. 5692-5707, Sept. 2013.

\bibitem{Wentu12}
W. Song, R. Feng, K. Cai and J. Zhang, ``Network Coding for
2-unicast with Rate (1,2)'' in \emph{Proc. ISIT}, Boston, 2012.
\end{thebibliography}
\end{document}